\documentclass[notitlepage,aps,pra,nofootinbib,twocolumn,floatfix,final,10pt]{revtex4-2}
\usepackage{tikz,tikz-cd}
\usepackage[ruled]{algorithm2e}
\usepackage{booktabs,multirow}
\usepackage{csquotes}\MakeOuterQuote"
\usepackage{pgfplots}\pgfplotsset{compat=1.18}

\usepackage[utf8]{inputenc}
\usepackage[T1]{fontenc}
\usepackage[english]{babel}
\usepackage{amsmath,amssymb,amsthm}
\usepackage{graphicx} % Required for inserting images
\usepackage[margin=1.in]{geometry}
\usepackage[shortlabels]{enumitem}
\usepackage{tikz,tikz-cd}
\usepackage{mathtools}
\usepackage{todonotes}
\usepackage{soul}\setstcolor{red}
\usepackage[colorlinks,citecolor=blue,linkcolor=blue]{hyperref}
\usepackage{cleveref}
\usepackage{comment}

\newtheorem{theorem}{Theorem}
\newtheorem{lemma}[theorem]{Lemma}
\newtheorem{corollary}[theorem]{Corollary}
\newtheorem{proposition}[theorem]{Proposition}

\newtheorem{question}{Question}
\newtheorem{definition}[theorem]{Definition}
\theoremstyle{definition}

\newtheorem*{problem*}{Problem}
\newtheorem*{assumption*}{Assumption}

\newtheorem*{warning*}{Warning}

\crefname{theorem}{Thm.}{Thms.}
\Crefname{theorem}{Theorem}{Theorems}
\crefname{proposition}{Prop.}{Props.}
\Crefname{proposition}{Proposition}{Propositions}
\crefname{lemma}{Lem.}{Lemmas}
\Crefname{lemma}{Lemma}{Lemmas}
\crefname{definition}{Def.}{Defs.}
\Crefname{definition}{Definition}{Definitions}
\crefname{equation}{eq.}{eqs.}
\Crefname{equation}{Equation}{Equations}
\crefname{figure}{fig.}{figs.}
\Crefname{figure}{Figure}{Figures}
\crefname{appendix}{Appendix}{Appendices}
\Crefname{appendix}{Appendix}{Appendices}

\newcommand{\ketbra}[2]{|#1\rangle\langle#2|}
\newcommand{\ket}[1]{|#1\rangle}
\newcommand{\kettbra}[1]{\ketbra{#1}{#1}}
\newcommand{\bra}[1]{\langle#1|}

\newcommand{\norm}[1]{\lVert #1\rVert}
\newcommand{\oo}{\infty}
\newcommand{\ox}{\otimes}
\newcommand{\mc}{\mathcal}
\newcommand{\eps}{\varepsilon}
\newcommand{\III}{{\mathrm{III}}}
\newcommand{\II}{{\mathrm{II}}}
\newcommand{\I}{{\mathrm{I}}}

\newcommand{\up}[1]{^{(#1)}}

\DeclareMathOperator{\tr}{Tr}
\DeclareMathOperator{\Tr}{Tr} %we need this command in the technical manuscript

\renewcommand{\tilde}{\widetilde}

\newcommand{\hide}[1]{}

\def\A{{\mc A}}

\def\B{{\mc B}}
\def\CC{{\mathbb C}}

\def\H{{\mc H}}
\def\K{{\mathcal K}}
\def\M{{\mc M}}
\def\N{{\mc N}}

\def\NN{{\mathbb N}}

\def\ZZ{{\mathbb Z}}

\newcommand{\R}{\mc R}

\newcommand{\Rep}{\mathrm{Rep}}
\newcommand{\Vect}{\mathrm{Vec}}

\def\locc{\xrightarrow{\LOCC}}

\newcommand{\LOCC}{\mathrm{LOCC}}

\newcommand{\placeholder}[0]{{\,\cdot\,}}

\def\locc{\xrightarrow{\LOCC}}
\def\lu{\xleftrightarrow{\,\mathrm{LU}\,}}

\usepackage{tikz}
\newcommand{\no}{%
\tikz[scale=0.23] {
    \draw[line width=0.7,line cap=round] (0.0,0.05) to [bend left=4] (.9,1);
    \draw[line width=0.7,line cap=round] (0.1,0.95) to [bend right=2] (0.8,0.05);
}}
\newcommand{\yes}{%
\tikz[scale=0.23] {
    \draw[line width=0.7,line cap=round] (0.25,0) to [bend left=10] (1,1);
    \draw[line width=0.8,line cap=round] (0,0.35) to [bend right=1] (0.23,0);
}}

\begin{document}

\title{The Large-Scale Structure of Entanglement in Quantum Many-body Systems}
\author{Lauritz van Luijk, Alexander Stottmeister, and Henrik Wilming}
\affiliation{\small Institut f\"ur Theoretische Physik, Leibniz Universit\"at Hannover, Appelstraße 2, 30167 Hannover, Germany}
\date{\today}

\begin{abstract}
We show that the thermodynamic limit of a many-body system can reveal entanglement properties that are hard to detect in finite-size systems -- similar to how phase transitions only sharply emerge in the thermodynamic limit. The resulting operational entanglement properties are in one-to-one correspondence with abstract properties of the local observable algebras that emerge in the thermodynamic limit. 
These properties are insensitive to finite perturbations and hence describe the \emph{large-scale structure of entanglement} of many-body systems.
We formulate and discuss the emerging structures and open questions, both for gapped and gapless many-body systems. In particular, we show that every gapped phase of matter, even the trivial one, in $D\geq 2$ dimensions contains models with the strongest possible bipartite large-scale entanglement. Conversely, we conjecture the existence of topological phases of matter, where all representatives have the strongest form of entanglement.
\end{abstract}
\maketitle

The thermodynamic limit is a powerful idealization to sharply characterize bulk properties of matter, such as their thermodynamic phases. 
In recent decades, entanglement properties of quantum many-body systems have been found to be crucial to understand their properties at zero and non-zero temperature, in particular their critical behavior \cite{osborne_entanglement_2002,vidal_entanglement_2003,calabrese_entanglement_2004,orus_half_2006} and topological order \cite{wen_topological_1990,wen_topological_1995,kitaevFaulttolerantQuantumComputation2003,kitaevAnyonsExactlySolved2006}. 
Discussing entanglement requires partitioning a system into multiple parts--the simplest scenario consisting of a bipartition.
So far, the study of entanglement in many-body systems is mostly restricted to settings where one of the parts is finite, resulting in necessarily finite amounts of entanglement, or studying the finite-size scaling of associated entanglement measures (see \cite{amicoEntanglementManybodySystems2008,eisertColloquiumAreaLaws2010} for reviews).

This poses the natural question whether it is operationally meaningful to talk about entanglement of quantum many-body systems directly in the thermodynamic limit, when each of the subsystems is infinite and thereby allowing for infinite entanglement. This question has a tradition in quantum field theory (QFT) and related settings, where infinitely many degrees of freedom are not necessarily an idealization; see \cite{hollands_entanglement_2018} for a recent account.
It is well known that entanglement measures, such as the entanglement entropy, diverge for all states in the vacuum sector of a relativistic QFT, indicating an infinite amount of entanglement that, for example, can be used to maximally violate Bell inequalities \cite{summers_vacuum_1985,summers_maximal_1987} or distill arbitrary entangled states \cite{verch_distillability_2005}.

Here, we invoke this question specifically for infinite-volume ground states, which are significantly less constrained than relativistic vacua due to the absence of Lorentz invariance. We ask: Do qualitatively new features appear in such a setting that are difficult or impossible to detect using finite-size scaling?
Can we find an (idealized) description for such settings?
In particular, do these properties have operational characterizations? And finally, do these properties also have operational significance for large but finite systems?
The aim of this paper is to give affirmative answers to these questions, explain the arising structures and formulate a number of emerging open problems.

\begin{figure}[t]
	\includegraphics[width=8cm]{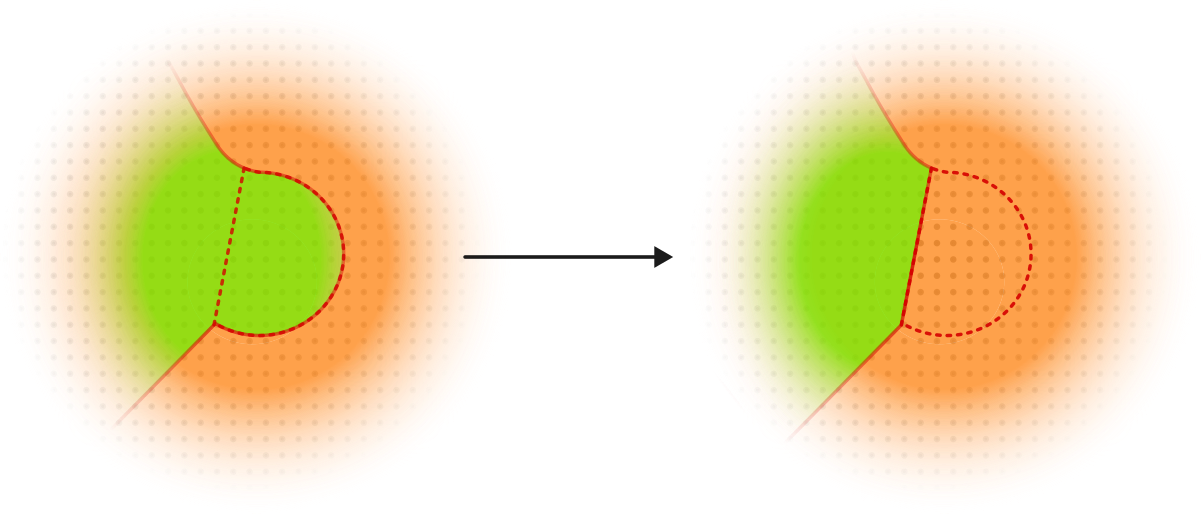}
	\caption{Illustration of \cref{prop:stable-type} and \cref{thm:stable-bipartite}: Modifying a bipartition into two (properly) infinite regions by a finite amount does not change the bipartite large-scale entanglement properties listed in \cref{tab:types}.}\label{fig:perturbation}
\end{figure}

Our results show that many-body systems can have strikingly different \emph{large-scale structures of entanglement}, which are insensitive to perturbations on arbitrarily large but finite subsystems and are in direct correspondence to operational tasks in entanglement theory. 
We relate these structures to critical systems, gapped phases of matter, and topological order. 

On a technical level, this paper collects and connects a number of results recently established in quantum information theory and many-body physics using operator algebraic techniques. In particular, we connect results relating properties of local observable algebras in quantum many-body systems to critical behavior or its absence (in one spatial dimension) \cite{matsui_split_2001,keyl_entanglement_2006,matsui_boundedness_2013,van_luijk_critical_2024} or topological order (in higher dimensions) \cite{naaijkensAnyonsInfiniteQuantum2012,fiedlerHaagDualityKitaevs2015,naaijkensQuantumSpinSystems2017,ogata_type_2022,naaijkensSplitApproximateSplit2022,jonesLocalTopologicalOrder2023,tombaBoundaryAlgebrasKitaev2023,bhardwaj_superselection_2024} with concrete operational tasks in entanglement theory. We hope that this article provides a starting point for researchers from quantum information theory, many-body theory, and mathematical physics alike to contribute to the understanding of the large-scale structure of entanglement in quantum many-body systems. 

\vspace{-0.25cm}
\subsection{Local operations and ground state sectors of many-body systems}
\vspace{-0.15cm}
Throughout, we are interested in entanglement properties of the ground states of quantum many-body systems in the thermodynamic limit. We now formulate the basic setting of this work. 
To be able to describe large-scale entanglement properties, we need to introduce some basic terminology for the operator algebraic description of many-body systems in the thermodynamic limit. For detailed treatments, see, for example, \cite{BR1,BR2}. 

We consider a many-body system consisting of spins localized at the sites of a lattice $\Gamma$, e.g., $\Gamma=\ZZ^D$.\footnote{Much of what follows can be formulated for fermionic systems, see, for example, \cite[Sec.~2.6]{BR1} as well as the Supplemental Information of \cite{van_luijk_critical_2024}.}
Each spin is described by the $d$-dimensional Hilbert space $\CC^d$.
For a finite region $A\Subset \Gamma$ ("$\Subset$" indicates finite subsets), we set $\H_A = \bigotimes_{x\in A}\CC^d$ and $\A_A = \B(\H_A)$.
If $A\subset B$, we identify $\A_A\subset\A_B$ as a subalgebra via the map $a_A \mapsto a_A \ox 1_{B\setminus A}$.
The union $\A_\Gamma = \bigcup_{A\Subset \Gamma} \A_A$ of this increasing net of subalgebras is the algebra of operators with finite support.\footnote{$\A_\Gamma$ is a unital $*$-algebra and comes with a natural operator norm, defined as $\norm a_{\A_\Gamma} = \norm a_{\B(\H_A)}$ for some/all $A\Subset \Gamma$ such that $a_A\in \A_A=\B(\H_A)$, which turns it into a pre-C*-algebra whose completion is the so-called quasi-local algebra \cite[Sec.~2.6.1]{BR1}.}
Even if a region $A\subset \Gamma$ is infinite, we can associate the subalgebra $\A_A = \bigcup_{A'\Subset A} \A_{A'}$ with it.

To connect the local operator algebras with quantum mechanics, we consider a Hilbert space $\H$ carrying an irreducible representation of $\A_\Gamma$. While distinguishing between $\A_\Gamma$ and its representation on a Hilbert space is irrelevant for finite systems, it is crucial for infinite systems. In an infinite system, different (inequivalent) representations may correspond to physically "orthogonal" settings, for example, different symmetry-broken phases. 
We explain below that they can also carry vastly different entanglement structures. 
Nevertheless, to simplify notation we will not distinguish between $\A_\Gamma$ and its representation as we will always work in a representation that arises from a ground state of a Hamiltonian, see below.
We refer to $\H$ as a \emph{sector} of the spin system.
The quantum states of the sector are given by the density operators $\rho$ on $\H$. \footnote{The sector $\H$ is fully determined by its states up to unitary equivalence. Mathematically speaking, this is because two irreducible representations are unitarily equivalent if and only if they induce the same sets of normal states \cite[Sec.~2.4.4]{BR1}.}
As always, we identify pure states with unit vectors $\Psi\in\H$ via $\rho = \kettbra\Psi$.

In the cases of interest, the sector $\H$ will be the ground state sector of a local Hamiltonian $H$, i.e., we can formally write $H= \sum_{A\Subset\Gamma} h_A$ for operators $h_A\in\A_A$ such that the norm $\norm{h_A}$ decays sufficiently fast as the region $A$ increases (interactions become weaker at larger distances); see \cite{BR2} for details.
If a pure ground state of the Hamiltonian $H$ is fixed, we can consider all the states arising by acting with operators from $\A_\Gamma$ and complete the resulting vector space to obtain the Hilbert space $\H$.
Mathematically speaking, we apply the Gelfand-Naimark-Segal (GNS) construction.
Then, $H$ becomes a positive self-adjoint operator on $\H$, and the ground state is realized by a specific vector $\Omega\in \H$ with $H\Omega=0$. Importantly, the resulting sector depends on the ground state, hence on the Hamiltonian, with potentially vastly varying entanglement properties.

To set up entanglement theory, we need to specify a notion of local operation. 
We consider operations that act on a \emph{finite} (but unbounded) number of spins. 
Such operations preserve the various sectors of the many-body system.
Operations localized in a region $A\subset\Gamma$ are specified by operators in $\A_A$.
To be more precise, a quantum channel $T_A$ localized in a region $A\subset\Gamma$ (possibly containing infinitely many spins) is given by Kraus operators $k_\alpha \in \A_A$ with $\sum_\alpha k_\alpha^*k_\alpha=1$, i.e.,
\begin{equation}
    \rho \mapsto T_A(\rho) = \sum_\alpha k_\alpha \rho k_\alpha^\dagger.
\end{equation}
Combining this notion of local operations with classical communication, we obtain a definition of local operations and classical communication (LOCC) with finite but unbounded support.
Thus, if $A_1,\ldots, A_N$ is a partition of $\Gamma$, then an LOCC protocol of the corresponding $N$-parties is an operation $T$ that can be obtained by iterating the basic step of a local operation in one of the regions followed by communicating all measurement results to all other regions (of course, the choice of the operation in the next step may depend on what has been communicated).
Accordingly, we say that a state $\sigma$ is approximately reachable from a state $\rho$ by means of LOCC (or that $\rho$ is more entangled than $\sigma$), denoted by
\begin{equation}
    \rho \locc \sigma,
\end{equation}
if for every $\eps>0$ there exists an LOCC operation $T$ such that $T(\rho)\approx_\eps\sigma$.\footnote{We write $\rho\approx_\eps \sigma$ if $\norm{\rho-\sigma}_1\le \eps$.}
Similarly, we write $\rho \lu \sigma$ if for all $\eps>0$, there exist unitaries $u_x \in \A_{A_x}$ such that $u\rho u^\dagger \approx_\eps \sigma$, where where $u=\prod_xu_x$.
If the transformation is only possible up to a fixed error $\eps$, we write $\locc_\eps$ or $\lu_\eps$.

For simplicity, we restrict our attention to LOCC transformations between pure states in the remainder of the paper. In case of pure states with representing vectors $\Psi,\Phi\in \H$, we write $\Psi\locc_{(\eps)}\Phi$ or $\Psi\lu_{(\eps)}\Phi$.
We can operationally distinguish different states by probing them with entanglement distillation tasks.

To describe entanglement distillation tasks, we associate to each subsystem $A_{j}$ of the $N$-partite many-body system a finite-dimensional quantum system with Hilbert space $\CC^n$ and observable algebra $\B_{j} \cong M_n(\CC)$. In other words, we attach to the $N$-partite many-body system an $N$-partite finite-dimensional quantum system with total Hilbert space $(\CC^n)^{\otimes N}$. 
On the extended system, we consider LOCC protocols where local quantum channels on the subsystem $j$ have Kraus operators $k_\alpha \in \A_{A_j}\ox \B_j$.
We always assume that the initial state on the auxiliary system is given by the product state $\ket{0}^{\otimes N}$ when considering distillation tasks.
Thus, all entanglement that is generated on the auxiliary system, represented by a final state $\ket\Phi \in (\CC^n)^{\otimes N}$, must be extracted from the many-body system at hand.

\begin{table*}[t]\centering
\setlength{\tabcolsep}{7pt}
\renewcommand\arraystretch{1.2}
\begin{tabular}{@{} l cc cc ccc @{}}
\toprule
\multirow{2}{*}[-0.5\dimexpr \aboverulesep + \belowrulesep + \cmidrulewidth]{operational property} & \multicolumn{2}{c}{type I} & \multicolumn{2}{c}{type II} & \multicolumn{3}{c}{type III}\\
\cmidrule(lr){2-3} \cmidrule(lr){4-5} \cmidrule(l){6-8}
& I$_n$       & \ I$_\infty$   & \ II$_1$    & \ II$_\infty$ & \ III$_0$ & \ III$_\lambda$ & \ III$_1$ \\ 
 \midrule
one-shot entanglement           & \!\!\! $\le$$\log_2n$\!    & \!$<$$\infty$    & $\oo$&$\oo$ & $\oo$&$\oo$&$\oo$                                \\
all pure states LOCC equivalent & \no&\no     &\no&\no      & \yes&\yes&\yes         \\ 
maximally entangled state       & \yes &    \no          & \yes        & \no             & \yes&\yes&\yes          \\
embezzling states             & \no&\no&\no&\no & (\,\yes\,) & \yes & \yes\\
worst embezzlement capability $\kappa_{\textit{max}}$ &  2 & 2    & 2&2       & 2 & \!\!$2\frac{1-\sqrt\lambda}{1+\sqrt\lambda}$\!\! & 0 \\ 
\bottomrule
\end{tabular}
\caption{Correspondence between operational entanglement properties in the ground state sector and the type classification of factors. In the table, $n$ denotes an integer $n\in\NN$ and $\lambda$ a real number $0<\lambda<1$.
The inequalities for the one-shot entanglement in type $\I$ are bounds on the maximal number of Bell states that can be distilled from any given state.
We refer to \cite{van_luijk_relativistic_2024,long_paper} for the formal definition of $\kappa_{\textit{max}}$ and the derivation of the formula for type $\III_\lambda$ factors. Since $\lambda\mapsto 2\frac{1-\sqrt\lambda}{1+\sqrt\lambda}$ is invertible, $\lambda$ is determined by the operational quantity $\kappa_{\textit{max}}$. 
If $\kappa_{\textit{max}}=0$ then every state is embezzling. Some $\III_0$ factors admit embezzling states, while others do not.}
\label{tab:types}
\end{table*}
\begin{definition}[{\cite{keyl_infinitely_2003}}]
    A state $\Phi\in (\CC^n)^{\ox N}$ can be distilled from $\Psi \in\H$ relative to fixed $N$-partition of $\Gamma$ if for every $\eps>0$, there exists an LOCC protocol $T_\eps$ such that
    \begin{equation}
        \bra\Phi \tr_\H T_\eps(\kettbra\Psi\ox \kettbra0^{\ox N}) \ket \Phi \geq 1-\eps.
    \end{equation}
\end{definition}
We caution the reader that $\Psi$ and $\Phi$ in the previous definition as well as below are elements of different Hilbert spaces: $\Psi\in \H$ is state of the many-body system while $\Phi\in (\CC^n)^{\ox N}$ is a finite-dimensional state shared among the local systems associated with the $N$-partition of $\Gamma$.

In this paper, we consider the following operational properties of a state vector $\Psi\in\H$ relative to a fixed $N$-partition of the lattice:
\begin{enumerate}
    \item \emph{Infinite one-shot entanglement}:
    Every $\Phi\in (\CC^n)^{\ox N}$ for every $n$ can be distilled from $\Psi$. 
    \item \emph{LOCC-embezzlement of entanglement}:
    For every $n\in\NN$, $\Phi\in (\CC^n)^{\ox N}$, 
    \begin{equation}
        \Psi \ox \ket0^{\ox N} \locc \Psi \ox \Phi.
    \end{equation}
    \item \emph{Embezzlement of entanglement}:
    For every $n\in\NN$, $\Phi\in (\CC^n)^{\ox N}$, 
    \begin{equation}
        \Psi \ox \ket0^{\ox N} \lu \Psi \ox \Phi.
    \end{equation}
\end{enumerate}
In the case of $N>2$ parties, these tasks demand the extraction of genuine multipartite entanglement.
Embezzlement of entanglement was discovered in an approximate form in finite-dimensional systems in \cite{van_dam_universal_2003}.
If the second or third property hold, we say that that $\Psi$ is an (LOCC-)embezzling state (or \emph{embezzler} for short).
Clearly, every embezzler is an LOCC-embezzler, and every LOCC-embezzler has infinite one-shot entanglement.

The fact that the definition of LOCC (and similarly that of LU) convertibility allows for arbitrarily small errors has the crucial consequence that we obtain the same notion of convertibility if we allow certain idealized operations. On the level of Kraus operators, one may pass to limits in the so-called strong operator topology (SOT).\footnote{The strong operator topology on $\B(\H)$ is the topology of pointwise convergence on $\H$: A net $(x_\alpha)$ converges to $x$ in the strong operator topology if and only if $\lim_\alpha \norm{(x_\alpha-x) \Psi}=0$ for all $\Psi\in\H$} 
The closure of the local algebra $\A_A$ in the SOT, is a von Neumann algebra
\begin{equation}
\M_A=\overline{\A_A}^{\text{SOT}}.
\end{equation}
A von Neumann algebra is a SOT-closed unital *-algebras of bounded operators on a Hilbert space.

Since $\A_A$ and $\A_B$ commute for disjoint regions $A\cap B=\emptyset$, the same holds for the von Neumann algebras $\M_A$ and $\M_B$. Thus, it follows that:

\begin{theorem}\label{thm:only factors matter}
    For a fixed partition of $\Gamma$ into regions $A_1,\ldots A_N\subset\Gamma$, all entanglement properties of the ground state sector $\H$ are encoded in the collection of commuting von Neumann algebras $\M_{A_1},\ldots \M_{A_N}$.
\end{theorem}

The importance of \cref{thm:only factors matter} is that it allows us to use von Neumann algebraic techniques to answer entanglement questions in many-body systems.\footnote{It is, however, not the case that $\Psi\locc\Phi$ implies that one can find an LOCC protocol with Kraus operators from the respective von Neumann algebras, which takes $\Psi$ to $\Phi$ with error $\eps=0$.}

Since the representation of $\A_\Gamma$ on $\H$ is irreducible, we have $\M_\Gamma=\B(\H)$.
Hence, the commuting algebras $\M_A$ and $\M_{A^c}$ generate $\B(\H)$ as a von Neumann algebra. This implies that the algebras $\M_A$ are so-called \emph{factors},\footnote{More precisely, the irreducibility of the representation on $\H$ implies that the inclusion $\M_{A}\subseteq\M_{A^{c}}'$ is an irreducible subfactor inclusion; see, for example, \cite[Ch.~9]{evans1998qsym}.} i.e., von Neumann algebras with $\M_A\cap \M_A'=\CC1$, where 
\begin{equation}
    \R'=\{a\in \B(\H) : [a,b]=0\ \forall b\in \R\}
\end{equation}
denotes the commutant of a collection of operations $\R\subset\B(\H)$.
Since the factors $\M_A$ are built from matrix algebras, they are what is called \emph{approximately finite-dimensional}\footnote{Also known as: hyperfinite, injective, amenable \cite{connes_classification_1976}.} factors (see \cite{takesaki3}).

\vspace{-0.25cm}
\subsection{Bipartite entanglement and von Neumann type}
\vspace{-0.15cm}
We now fix a bipartition $\Gamma=A\cup A^c$ of the latticeand assume that \emph{Haag duality} holds, which means that $\M_A$ and $\M_{A^c}$ are commutants of each other:
\begin{equation}
    \M_A' = \M_{A^c}.
\end{equation}
Haag duality means that every symmetry of $A$, i.e., any unitary operator $u$ such that $u a u^\dagger = a$ for all $a\in\M_A$, is an element of $\M_{A^c}$. By \cref{thm:only factors matter}, what we say in the following actually holds for any bipartite system described by commuting factors on a Hilbert space $\H$ in Haag duality together with any pair of weakly dense *-subalgebras $\A_{A\up c}\subset\M_{A\up c}$.
Taken together, the results of the recent papers \cite{van_luijk_pure_2024,long_paper,van_luijk_relativistic_2024} show that the bipartite entanglement structure is in one-to-one correspondence with the algebraic structure of the factors $\M_A$ and $\M_{A^c}$.
These results are summarized in \cref{tab:types}. To explain them, we briefly sketch the classification of factors in von Neumann algebras theory:

At a coarse level, factors are sorted into three classes, called type $\I$, $\II$ and $\III$.
Type $\I$ factors are of the form $\B(\K)$ for some Hilbert space $\K$ and are sub-classified into types $\I_n$ where $n=\dim\K$ (which may be infinite).
Type $\II$ factors are those factors that are not of type $\I$ but still admit a nontrivial trace. They come in types $\II_1$ and $\II_\oo$ (depending on whether the identity operator has a finite trace or not). All remaining factors are of type $\III$. Type $\III$ factors are classified into subtypes $\III_\lambda$, $0\le \lambda\le 1$. The classification of approximately finite-dimensional factors is a deep mathematical result, which was essentially completed by the works of Connes and Haagerup \cite{connes1973classIII, connes_classification_1976,haagerup_connes_1987}.
With the exception of type $\III_0$ factors, all approximately finite-dimensional factors (appearing in our case, see above) are completely determined by the type classification; see, for example, \cite{takesaki2,takesaki3} for a detailed exposition.
The results shown in \cref{tab:types} show that type and subtype determine and are determined by the operational tasks discussed above and whether they allow for maximally entangled states.

For instance, a type $\I$ ground state sector is equivalent to the ground state
having finite one-shot entanglement, a type $\III$ sector is equivalent to the ground state being LOCC-embezzler, and a type $\II_\oo$ ground state sector is equivalent to infinite one-shot entanglement and the non-existence of a maximally entangled state.

Entanglement is often considered in a \emph{resource theory}: Alice and Bob may prepare arbitrary systems local to them in arbitrary pure states and apply LOCC transformations "for free."\footnote{In fact, by allowing for general localized quantum channels, we effectively already allow each of the agents, controlling a region $A_i$, to locally couple their subsystem to a quantum mechanical system prepared in an arbitrary pure state.}
Entangled states, mixed or pure, correspond to those states that cannot be prepared for free in this sense.
Every entangled state is a \emph{resource state}, from which other entangled states may be distilled via LOCC, and the amount of entanglement it contains can be measured by the states that can be distilled from it.
Given a bipartite system described by $\M_A,\M_B$ on $\H$ fulfilling Haag duality, product states relative to $A$ and $B$ exist in $\H$ if and only if $\M_A$ has type $\I$.
In case of types $\II$ and $\III$, all states, even mixed ones, are entangled. 
Indeed, comparing to \cref{tab:types}, we find that all (pure) states have infinite one-shot entanglement and hence allow to distill arbitrary entangled states.
Nevertheless, in the type $\II$ case, there is still a non-trivial LOCC ordering between entangled states, i.e., a non-trivial resource theory, and we can define entanglement measures, such as entanglement entropies, in a meaningful way, see for example \cite{crann_state_2020,van_luijk_pure_2024}.
On the other hand, the LOCC resource theory for pure states trivializes in type $\III$ case: Any two pure states can be interconverted via LOCC. We will see below that all these scenarios can arise in ground state sectors of many-body systems.

\vspace{-0.25cm}
\subsection{Stability to perturbations}
\vspace{-0.15cm}
The entanglement properties listed in \cref{tab:types} are properties of the full ground state sector and not of individual states. 
Crucially, under the assumption of Haag duality these do not change if a finite number of degrees of freedom are moved from Alice's part to Bob's (see \cref{fig:perturbation}).

\begin{proposition}\label{prop:stable-type}
    If a region $A$ is infinite (contains infinitely many degrees of freedom), then $\M_A$ is isomorphic to $\M_{A\cup R}$ whenever $R$ is finite. Moreover, if Haag duality holds  for $A$ (i.e., $\M_A' = \M_{A^c}$) it holds holds for $A\cup R$.
\end{proposition}
The essential ingredient for the proposition is that approximately finite-dimensional factors not of type $\I_n$ with $n<\infty$ are \emph{stable}, meaning that
\begin{align}
	\M \cong \M \ox \B(\K)
\end{align}
for any finite-dimensional Hilbert-space $\K$, see \cref{app:stability}.
A stronger statement can be made in the case that there exists a translation invariant pure state in the sector and that both $A$ and $A^c$ are \emph{properly infinite}, meaning that each contains balls of arbitrary size. 

\begin{theorem}\label{thm:stable-bipartite}
    Suppose $\H$ contains a translation invariant vector $\Omega$.
    Let $\Gamma = A\cup B$ be a bipartition into properly infinite regions for which Haag duality holds.
    Let $R\Subset A$ be finite.
    If we set $\tilde A = A\setminus R$ and $\tilde B=B\cup R$,
    there exists a unitary $U$ on $\H$ such that
    \begin{align}
        \begin{aligned}U\M_{A} U^\dagger &= \M_{\tilde A},\\  U\M_{B} U^\dagger &= \M_{\tilde B}.
        \end{aligned}
    \end{align}
\end{theorem}
\Cref{prop:stable-type,thm:stable-bipartite} show that the entanglement properties of a sector listed in \cref{tab:types} are insensitive to finite perturbations consisting of either changing the bipartition in a finite way. Moreover, these entanglement properties are, by definition, not affected by quantum channels acting on any finite region.
We may, therefore, say that they characterize the bipartite large-scale structure of entanglement of a many-body system. 
In the remainder, we explore these properties in one and two spatial dimensions.

The operational tasks associated with the type classification currently crucially rely on Haag duality, which must be verified in concrete models because there is no general argument that shows Haag duality in spin systems, or be replaced by a suitably weakened version as, for example, in the context of topological order \cite{ogata_derivation_2022}.
An important open question, therefore, is:
\begin{question}
    Can Haag duality fail in ground states of locally interacting spin systems for simply connected regions?
    If so, what operational significance does the failure of Haag duality have for operational entanglement properties of pure ground state?
\end{question}
As noted in the context of \cref{thm:only factors matter}, Haag duality does not fail in an arbitrary sense for pure ground states as a bipartition $\Gamma = A\cup A^{c}$ necessarily leads to an irreducible subfactor inclusion, i.e., $\M_{A}\subseteq\M_{A^{c}}'$, such that $\M_{A}'\cap\M_{A^{c}}'=\CC$.

It is known that Haag duality fails for \emph{disconnected} subsystems consisting of two spatially separated infinite cones in two-dimensional topologically ordered systems due to the existence of ribbon operators connecting anyons that can be deformed freely \cite{naaijkensLocalizedEndomorphismsKitaevs2011,naaijkensAnyonsInfiniteQuantum2012,fiedlerHaagDualityKitaevs2015}.
In \cite{fiedlerJonesIndexSecret2017}, the authors relate this observation to secret sharing schemes arising in topologically ordered systems in terms of a \emph{tripartition} of the lattice. 
The authors further discuss an operational interpretation of the index of finite-index subfactor inclusions $\N \subset\M$ in terms of the amount of information that can be accessed from $\M$ as compared to $\N$, see also \cite{gaoRelativeEntropyNeumann2020}. However, both interpretations do not directly lend an operational interpretation to the failure of Haag duality in terms of bipartite entanglement theory for observers associated with subsystems $A$ and $A^c$, where, in general, we cannot associate a physical system with the factor $\M_A'$ that strictly contains $\M_{A^c}$. 

\vspace{-0.25cm}
\subsection{Finite systems}
\vspace{-0.15cm}
The large-scale properties of entanglement only have a sharp description in the thermodynamic limit.
What effect do the different types of von Neumann algebras have for large, but finite systems?
In essence, the operational tasks that can be achieved in the thermodynamic limit can be achieved approximately for any desired precision $\eps>0$ if the system is sufficiently large. 
This is due to the fact that the LOCC operations we consider only act on finite subsystems, the sector $\H$ consists of (quasi-)local excitations of the ground state, and operators localized in finite subsystems cannot distinguish well the ground state in the thermodynamic limit from the ground state in sufficiently large systems (if the ground state in the thermodynamic limit is unique). 

We illustrate this in two ways. First, suppose that a local many-body system has a unique ground state $\Omega$ in the thermodynamic limit and denote by $\Omega_L$ the sequence of ground states when the system is defined on concentric balls $B(L)$ of radius $L$ (with open boundary conditions). Furthermore, fix some partition $\Gamma = \cup_{i=1}^N A_i$ of the infinite lattice. For each system-size $L$ we can then consider the $N$-partite system made up of parts $A_i\up{L} := A_i \cap B(L)$. Then the following is true:

\begin{proposition}\label{prop:distillation}
    If $\Psi$ can be distilled from $\Omega$ via LOCC, then for every $\eps>0$ there exists a system-size $L_0$ such that $\Psi$ can be distilled from $\Omega_L$ up to error $\eps$ for all $L\geq L_0$.
\end{proposition}
The proof is presented in \cref{app:proof-distillation}, which, in addition, shows that the same LOCC protocol can be used for all system sizes $L\geq L_0$.

Thus, if the sector contains infinite one-shot entanglement, the sequence $\{\Omega_L\}$ contains unbounded one-shot entanglement as well.
Our second illustration comes from embezzlement. It shows that if the ground state $\Omega$ is an embezzling state, then $\{\Omega_L\}$ induce an embezzling family in the sense of \cite{van_dam_universal_2003,van_luijk_multipartite_2024}:

\begin{proposition}[\cite{van_luijk_critical_2024,van_luijk_multipartite_2024}]\label{prop:finite-size}
	Suppose $\Omega$ is an embezzling state. Then $\{\Omega_L\}$ induces an \emph{embezzling family}: There exists a sequence of system sizes $L_{K}$ such that for any $\eps>0$ and any $n$ there exists an $K_*(n,\eps)$ such that for all $K\geq K_*(n,\eps)$ and any two unit vectors $\Psi,\Phi \in (\CC^n)^{\ox N}$ we have
	\begin{align}
	\rho_K \ox \kettbra{\Psi} \lu_\eps \rho_K \ox \kettbra{\Phi},
	\end{align}
where $\rho_K = \tr_{B(K)^c}\kettbra{\Omega_{L_K}}$ is the reduced state of $\Omega_{L_K}$ on $B(K)$,  partitioned into subsystems $A_i\up{K}$.
\end{proposition}

In \cref{prop:finite-size}, we assume that the partitioning of the infinite lattice is fixed a priori, which implies that we have to consider the, in general mixed, reduced states $\rho_K$ as embezzling family relative to the 
subsystems $A_i\up{K}$.
We can consider the family of pure states $\Omega_{L_K}$ itself as a pure embezzling family, if we allow one of the agents, say for $i=1$, to also act on the  relative complement of the ball $B(K)$, so that $B(L_K)$ is partitioned into $\tilde A_1\up{K} = A_1\up{K} \cup(B(L_K)\setminus B(K))$ and $\tilde A_i\up{K}=A_i\up{K}$ for $i>1$. 
This is a straightforward consequence of the fact that purifications are unique up to unitaries on the purifying subsystem.
Let us now turn to the question of how the large-scale entanglement properties relate to other physical properties of the system.

\vspace{-0.25cm}
\subsection{Ground states in $D=1$}
\vspace{-0.15cm}
A natural class of models to study bipartite entanglement in systems with unboundedly many degrees of freedom is given by translation-invariant 1D spin chains partitioned into left and right half chains.

\begin{question}\label{q:entropy growth}
    Which types can occur for translation-invariant spin chains?
    Is the type determined by the growth rate of the entanglement entropy alone?
\end{question}

Combined results of Matsui and Hastings \cite{hastings_area_2007, matsui_boundedness_2013, ukai_hastings_2024} show that gapped ground states yield type $\I$ sectors and, hence, no interesting large-scale entanglement properties. 
The converse, however, does not hold, i.e., there are translation-invariant pure states with type $\I$ sectors and an unbounded entanglement entropy in 1d spin chains \cite{ReinhardPrivate}, although we presently do not know whether they can appear as ground states of local Hamiltonians.\footnote{Werner's construction of translation invariant pure states with diverging entanglement entropy and type I sectors is based on the idea of finitely correlated states \cite{fannesFinitelyCorrelatedStates1992} (or matrix product states) but with infinite bond dimension \cite{ReinhardPrivate}.}

Therefore, we now consider non-gapped models in 1D.
We have shown in \cite{van_luijk_critical_2024} that all critical translation-invariant, and free (Gaussian) fermion models satisfy Haag duality and have type $\III_1$ sectors.
It has been argued using CFT arguments \cite{calabrese_entanglement_2008} that all critical spin systems in one spatial dimension have a universal distribution of Schmidt coefficients. This suggests that all critical spin chains, i.e., those with a CFT scaling limit, have type $\III_1$ sectors.

It is known that type $\II_1$ cannot occur \cite{keyl_entanglement_2006}, and we currently expect that the same holds for type $\II_\oo$.
On the other hand, the "colored Motzkin spin chains" and its relatives  \cite{movassaghSupercriticalEntanglementLocal2016,salbergerFredkinSpinChain2017,salbergerDeformedFredkinSpin2017,zhangQuantumPhaseTransition2016a,menonSymmetriesCorrelationFunctions2024} are gapless and show a \emph{supercritical} growth of entanglement entropy as $O(\sqrt{L})$ or quicker instead of $O(\log L)$ with the systems size $L$, which corresponds to considerably stronger finite-size entanglement than in critical systems.\footnote{While the colored Motzkin spin chain is not translation-invariant on a finite system, it does yield a translation-invariant state in the thermodynamic limit.} In this respect, we note that the original van Dam-Hayden embezzling family \cite{van_dam_universal_2003} shows an $O(L)$-growth of the entanglement entropy if realized via spin chains of length $L$. But, as argued in \cite{van_luijk_multipartite_2024}, said family leads to a $\II_{1}$ factor for the half-chain, which is possible because the model is not translation invariant.
It will be interesting to see whether the Motzkin models also yield type $\III_1$ sectors or whether they allow for different types, which would disprove a conjecture from \cite{keyl_entanglement_2006} that only type $\I$ and $\III_1$ can appear. 

A second open question concerns Haag duality in $D=1$ spin chains.
\begin{question}\label{q:HD}
    Does Haag duality automatically hold for translation-invariant pure states on spin chains and half-chain bipartitions?
\end{question}
An affirmative answer was claimed in \cite{keyl2008haag} but their proof contains an error.\footnote{This is noted by Matsui in \cite{matsui_boundedness_2013} and by Werner in \cite{van_luijk_pure_2024}.}
By \cite{van_luijk_schmidt_2024}, an affirmative answer would imply that all "finitely correlated" pure states are "C*-finitely correlated" \cite{fannesFinitelyCorrelatedStates1992}, i.e., matrix-product states \cite{perez-garciaMatrixProductState2007}.

\begin{figure}
    \centering
    \includegraphics[width=.7\linewidth]{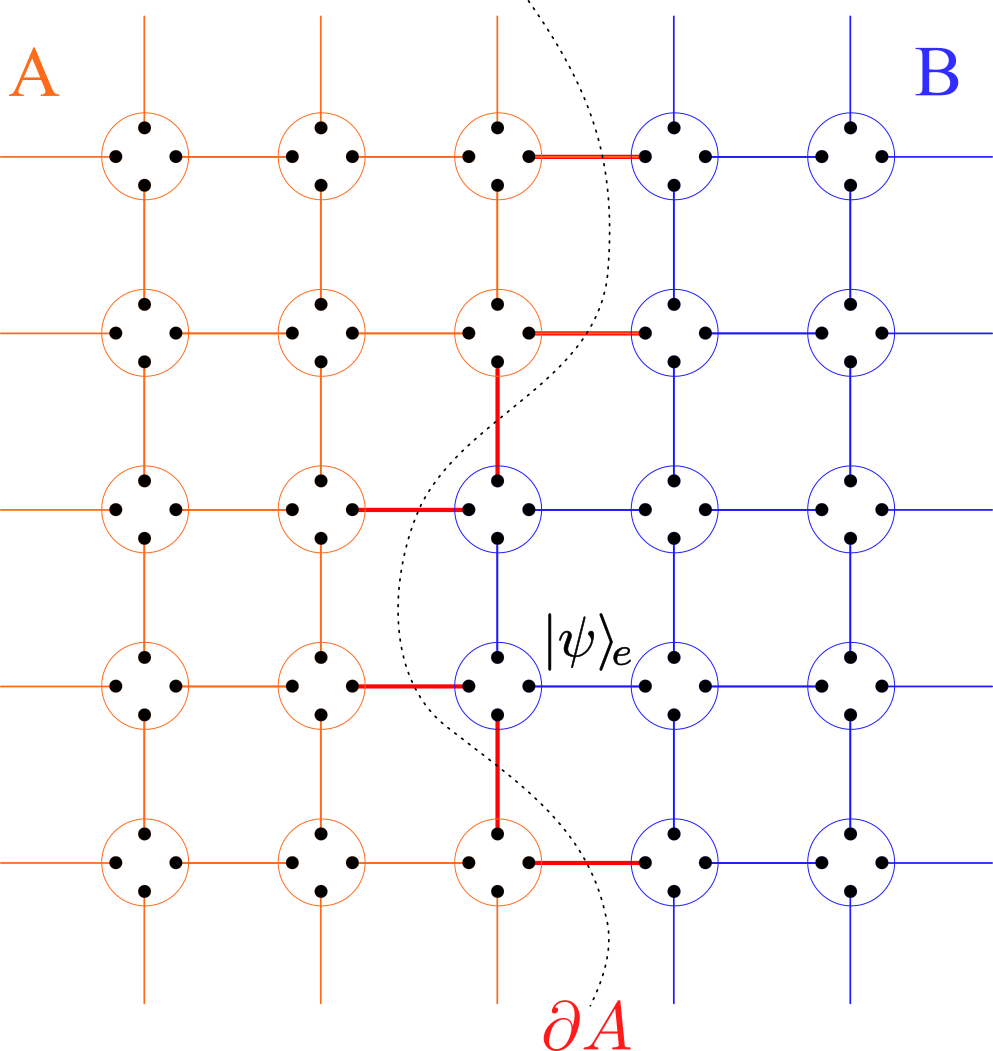}
    \caption{The basic construction. Each circle corresponds to a lattice site $v\in \mc V$ in the vertex set $\mc V$ of a graph $\mc G = (\mc V,\mc E)$ with constant degree $d$ (here, $d=4$).
    Each site $v$ is subdivided into $d$ virtual spins $\tilde v$ of dimension $m\geq 3$. 
    Each edge $e\in\mc E$ is associated with a neighboring pair of virtual spins and a copy of a canonical purification  $\ket\psi\in \CC^m\ox \CC^m$ of a density matrix $\rho$ is placed on each such edge. 
    By construction, the reduced state of every site is thus given by $\rho^{\otimes d}$.
    When the lattice is subdivided into two infinite parts $A$ and $B$, the two parts are only entangled through the infinite tensor product $\ket\psi^{\ox \infty}$ along the boundary, determining the type of local factors $\M_A,\M_B$.}
    \label{fig:construction}
\end{figure}

\vspace{-0.25cm}
\subsection{Gapped ground states in $D\ge 2$}
\vspace{-0.15cm}
Two gapped, local Hamiltonians are said to be in the same (gapped) phase of matter if they can be connected by a path of gapped, local Hamiltonians with a uniform lower bound on the gap along the path. 
The trivial phase corresponds to those Hamiltonians connected to a non-interacting Hamiltonian whose unique ground state is an infinite product state of arbitrary local dimension.
Moreover, if $H_0$ is in the trivial phase, then $H\ox 1+1\ox H_0$ is in the same phase as $H$.\footnote{However, $H_1\ox 1+ 1\ox H_2$ may be trivial even if $H_1,H_2$ are both non-trivial \cite{kitaevAnyonsExactlySolved2006}.}
We refer to this operation as "stacking."

In figure~\ref{fig:construction}, we present a very simple construction, which takes as input an $m$-dimensional density matrix $\rho$ and a graph $\mc G=(\mc V,\mc E)$ of constant degree $d$  (for example a hypercubic lattice in $D$ dimensions) and yields a local, frustration-free, gapped, commuting projector Hamiltonian
\begin{align}
    H^\rho = - \sum_{e\in \mc E} P^{\rho}_e
\end{align}
with a unique ground state. If $\mc G$ is a regular graph, it is also translation-invariant. The construction ensures that the Hamiltonian $H^\rho$ is in the trivial phase, independent of the choice of $\rho$.

Nevertheless, in spatial dimension $D\geq 2$, varying $\rho$ yields all of the types $\I_\oo,\II_\oo,\III_\lambda$ with $0<\lambda\leq 1$ for infinite subsystems while respecting Haag duality. Since stacking with a trivial Hamiltonian does not change the phase, we conclude:
\begin{corollary}\label{cor:III1}
Type $\III_1$ appears in every gapped phase of matter if $D\geq 2$. Every gapped phase of matter in $D\geq 2$ contains models with an embezzling ground state.
\end{corollary}

In fact, every phase that contains a model satisfying Haag duality also contains a model that is a universal embezzler.
To see this, take the model for which Haag duality holds and stack it with the topologically trivial model $H^\rho$, with $\rho$ chosen so that the resulting type of $\M_A$ is $\III_1$.
If Haag duality holds for the original system, it will also hold for the stacked system, and we can apply the arguments from \cite{long_paper}. 
Haag duality in the original model is necessary to ensure that all states are embezzling. Even if Haag duality fails for the original model, it still follows that all states that are uncorrelated between the two stacked subystems are embezzling.
\vspace{-0.25cm}
\subsection{Long-range vs large-scale entanglement}
\vspace{-0.15cm}
Corollary~\ref{cor:III1} shows that embezzlement neither requires criticality nor does it require long-range entanglement (i.e., a non-trivial gapped phase). Comparing with table \ref{tab:types}, we can order infinite factors according to their types as,
\begin{align}
    \I_\oo<\II_\oo < \III_0 < \III_\lambda < \III_1,
\end{align}
reflecting their resourcefulness for LOCC operations. 
According to this ordering, the maximum amount of large-scale entanglement (corresponding to type $\III_1$) is the same in each gapped phase. 
However, the minimum amount of large-scale entanglement does depend on the phase. 
For example, in two spatial dimensions, it was shown in \cite{naaijkensSplitApproximateSplit2022} that \emph{topological order} requires $\M_A$ not to be of type $\I$ if $A$ is properly infinite. 
For Levin-Wen (LW) models \cite{levinStringnetCondensationPhysical2005}, representing non-chiral topologically ordered phases, it was found in \cite{jonesLocalTopologicalOrder2023} that cone-shaped subsystems have associated factors of type $\III$ if and only if one of the simple objects $i$ in the unitary fusion category $\mc C$, from which the model is constructed, has quantum dimension $d_i\neq 1$.
In fact, it was very recently shown that the type may be computed precisely from the fusion rules and quantum dimensions \cite[App. A]{jonesLocalTopologicalOrder2023}.

It is natural to expect that stable renormalization group (RG) fixed points, such as LW models, correspond to those representatives of a gapped phase with the minimal amount of large-scale entanglement. However, care must be taken:
Different LW models may correspond to the same gapped phase and bulk anyon content, which is expected to be described by the \emph{Drinfeld center} $Z(\mc C)$ of the unitary fusion category $\mc C$, which is a unitary modular tensor category (UMTC) \cite{levinStringnetCondensationPhysical2005,kirillov_jr_string-net_2011,kitaevModelsGappedBoundaries2012,bultinck2017anyons,green_enriched_2024}. For example, if $G$ is a finite group, then $Z(\Vect_G) = Z(\Rep_G)$. One says that $\Vect_G$ and $\Rep_G$ are \emph{Morita equivalent} \cite{muger_subfactors_2003,lootens_mapping_2022}. 
The LW models associated with Morita equivalent categories correspond to the same phase and bulk anyon content, but may have different boundary behavior \cite{kitaevModelsGappedBoundaries2012}. While $\Vect_G$ always yields type $\II_\infty$, $\Rep_G$ yields type $\III$ for non-abelian $G$. 
Thus, the LW models for $\Rep_G$ do not have minimal large-scale entanglement.
In particular, have the following open problem:
\begin{question}
	Does there exist a gapped phase in $D=2$ where all representatives yield type $\III_{\lambda}$ algebras with $\lambda\geq \lambda_*$ for properly infinite subsystems? Is there a phase with $\lambda_* =1$?
\end{question}
Specializing to RG fixed points represented by LW models, a minimal requirement for such a phase would be to find a Morita equivalence class of unitary fusion categories, where all representatives lead to type $\III_{\lambda}$ algebras for some $\lambda\in(0,1]$. That such examples exist can be inferred from the fact that the total quantum dimension $D(\mc C) = \sqrt{\sum_{i}d_{i}^{2}}$ is invariant under Morita equivalence \cite{etingof_fusion_2005}:\footnote{Note that $D(\mc C)\ge1$ for a fusion category over $\CC$ and $D(Z(\mc C))=D(\mc C)^{2}$.} The total quantum dimension of the well-known Fibonacci fusion category $\textup{Fib}$ \cite{wang_topological_2010,natale_classification_2019} is $D(\textup{Fib}) = \sqrt{1+\phi^2}$, $\phi$ being the golden ratio, which forces one of the simple objects to a have a quantum dimension $d_{i}\neq1$. In particular, it is known that $\textup{Fib}$ is the unique representative of its Morita equivalence class, which entails $\lambda_{\textup{Fib}}=\phi^{-1}$ for the associated LW fixed-point model using the results of \cite{jonesLocalTopologicalOrder2023}.
More generally, RG fixed points, which are represented by LW models built from fusion categories that are not Morita equivalent to a pointed fusion category \cite{etingof_tensor_2015}, yield candidates for gapped phases with non-trivial $\lambda_{*}$. A second example, besides $\textup{Fib}$, is the Ising fusion category \cite{wang_topological_2010, natale_classification_2019}, leading to $\lambda_{\textup{Ising}}=1/2$ using the results of \cite{jonesLocalTopologicalOrder2023}.
By stacking a LW model for $\textup{Fib}$ with one for $\textup{Ising}$ we obtain a model with local algebras of type $\III_1$, since the types of approximately finite-dimensional factors fulfill $\III_\lambda \ox \III_\mu \cong \III_1$ if $\log(\lambda)/\log(\mu)$ is irrational \cite[Thm.~XVIII.4.16]{takesaki3}.\footnote{For any $m\in \NN$, the golden ratio $\phi$ fulfills $\phi^m = k\phi + n$ with $k,n\in\NN$. Consequently $\log(\phi)/\log(2)$ being rational would imply that $\phi$ is rational, a contradiction.}
We conjecture that the resulting model is part of a gapped phase where properly infinite subsystems have type $\III_1$ in ground state sectors of all representative models and the bulk anyon content is described by the UMTC $Z(\textup{Fib}\boxtimes \textup{Ising})$, where $\boxtimes$ is the Deligne tensor product \cite{etingof_tensor_2015}, which corresponds to a stacking of the associated LW models. If Haag duality holds throughout the phase, this would imply that all models in the phase are universal embezzlers.
\vspace{-0.25cm}
\subsection{Outlook and discussion}
\vspace{-0.15cm}
Studying the entanglement properties of ground states of local Hamiltonians has been extremely useful, both for conceptual understanding, but also for the development of numerical algorithms based on tensor networks \cite{schollwoeckDensitymatrixRenormalizationGroup2011,ciracMatrixProductStates2021}. So far, the study of entanglement in ground states was almost completely restricted to finite systems.
Employing recent abstract results from operator algebraic entanglement theory, we are now in a position to discuss the different kinds of infinite entanglement that can occur in quantum many-body systems and relate them to operational properties. 
The next step is to find out how large-scale entanglement properties relate to other physical properties of many-body systems. We have listed above only some of the pressing open questions. Of course, there are many more. For example, it would be interesting to better understand the influence of symmetries or their spontaneous breaking. For instance, in certain settings, spontaneous symmetry breaking implies infinite one-shot entanglement by preventing type $\I$ factors \cite{matsui_split_2001} (assuming Haag duality holds).
A second avenue for further research that we did not discuss above is to explore multipartite entanglement properties. In \cite{van_luijk_multipartite_2024}, we have constructed examples of multipartite embezzling states for an arbitrary number of parties.
This suggests to ask whether such states can appear as ground states of many-body systems:
\begin{question}
Do there exist many-body systems whose ground states contain infinite $N$-partite one-shot entanglement for $N\geq 3$? 
If so, are there ground states that are multipartite embezzling?
\end{question}
Recently, the possible types of von Neumann algebras describing the observable accessible to localized observers in spacetime have attracted considerable interest in the context of quantum gravity \cite{leutheusser2023emergent,witten2022crossed_product,chandrasekaran_algebra_2023,longo2023continuous_entropy,akers_relative_2024,jensen_generalized_2023,kudler-flam_covariant_2023,chenClockJustWay2024} and quantum reference frames \cite{fewsterQuantumReferenceFrames2025,devuystGravitationalEntropyObserverdependent2024,vuystLinearizationInstabilitiesCrossed2024,vuystCrossedProductsQuantum2024}. The results summarized in this work show that essentially all possible types of bipartite entanglement structures can already be found in quantum many-body systems, even those described by simple tensor networks. This may be relevant for toy models of quantum gravity and holography in terms of many-body systems.
\vspace{-0.25cm}
\begin{acknowledgments}
\vspace{-0.15cm}
We thank  Albert H. Werner, Reinhard F. Werner, and Laurens Lootens for interesting discussions and useful correspondence. 
LvL and AS have been supported by the MWK Lower Saxony via the Stay Inspired Program (Grant ID: 15-76251-2-Stay-9/22-16583/2022).
\end{acknowledgments}

\bibliographystyle{apsrev4-2}
\bibliography{refs.bib}

%apsrev4-2.bst 2019-01-14 (MD) hand-edited version of apsrev4-1.bst
%Control: key (0)
%Control: author (72) initials jnrlst
%Control: editor formatted (1) identically to author
%Control: production of article title (-1) disabled
%Control: page (0) single
%Control: year (1) truncated
%Control: production of eprint (0) enabled
\begin{thebibliography}{89}%
\makeatletter
\providecommand \@ifxundefined [1]{%
 \@ifx{#1\undefined}
}%
\providecommand \@ifnum [1]{%
 \ifnum #1\expandafter \@firstoftwo
 \else \expandafter \@secondoftwo
 \fi
}%
\providecommand \@ifx [1]{%
 \ifx #1\expandafter \@firstoftwo
 \else \expandafter \@secondoftwo
 \fi
}%
\providecommand \natexlab [1]{#1}%
\providecommand \enquote  [1]{``#1''}%
\providecommand \bibnamefont  [1]{#1}%
\providecommand \bibfnamefont [1]{#1}%
\providecommand \citenamefont [1]{#1}%
\providecommand \href@noop [0]{\@secondoftwo}%
\providecommand \href [0]{\begingroup \@sanitize@url \@href}%
\providecommand \@href[1]{\@@startlink{#1}\@@href}%
\providecommand \@@href[1]{\endgroup#1\@@endlink}%
\providecommand \@sanitize@url [0]{\catcode `\\12\catcode `\$12\catcode
  `\&12\catcode `\#12\catcode `\^12\catcode `\_12\catcode `\%12\relax}%
\providecommand \@@startlink[1]{}%
\providecommand \@@endlink[0]{}%
\providecommand \url  [0]{\begingroup\@sanitize@url \@url }%
\providecommand \@url [1]{\endgroup\@href {#1}{\urlprefix }}%
\providecommand \urlprefix  [0]{URL }%
\providecommand \Eprint [0]{\href }%
\providecommand \doibase [0]{https://doi.org/}%
\providecommand \selectlanguage [0]{\@gobble}%
\providecommand \bibinfo  [0]{\@secondoftwo}%
\providecommand \bibfield  [0]{\@secondoftwo}%
\providecommand \translation [1]{[#1]}%
\providecommand \BibitemOpen [0]{}%
\providecommand \bibitemStop [0]{}%
\providecommand \bibitemNoStop [0]{.\EOS\space}%
\providecommand \EOS [0]{\spacefactor3000\relax}%
\providecommand \BibitemShut  [1]{\csname bibitem#1\endcsname}%
\let\auto@bib@innerbib\@empty
%</preamble>
\bibitem [{\citenamefont {Osborne}\ and\ \citenamefont
  {Nielsen}(2002)}]{osborne_entanglement_2002}%
  \BibitemOpen
  \bibfield  {author} {\bibinfo {author} {\bibfnamefont {T.~J.}\ \bibnamefont
  {Osborne}}\ and\ \bibinfo {author} {\bibfnamefont {M.~A.}\ \bibnamefont
  {Nielsen}},\ }\href {https://doi.org/10.1103/PhysRevA.66.032110} {\bibfield
  {journal} {\bibinfo  {journal} {Physical Review A}\ }\textbf {\bibinfo
  {volume} {66}},\ \bibinfo {pages} {032110} (\bibinfo {year}
  {2002})}\BibitemShut {NoStop}%
\bibitem [{\citenamefont {Vidal}\ \emph {et~al.}(2003)\citenamefont {Vidal},
  \citenamefont {Latorre}, \citenamefont {Rico},\ and\ \citenamefont
  {Kitaev}}]{vidal_entanglement_2003}%
  \BibitemOpen
  \bibfield  {author} {\bibinfo {author} {\bibfnamefont {G.}~\bibnamefont
  {Vidal}}, \bibinfo {author} {\bibfnamefont {J.~I.}\ \bibnamefont {Latorre}},
  \bibinfo {author} {\bibfnamefont {E.}~\bibnamefont {Rico}},\ and\ \bibinfo
  {author} {\bibfnamefont {A.}~\bibnamefont {Kitaev}},\ }\href
  {https://doi.org/10.1103/PhysRevLett.90.227902} {\bibfield  {journal}
  {\bibinfo  {journal} {Physical Review Letters}\ }\textbf {\bibinfo {volume}
  {90}},\ \bibinfo {pages} {227902} (\bibinfo {year} {2003})}\BibitemShut
  {NoStop}%
\bibitem [{\citenamefont {Calabrese}\ and\ \citenamefont
  {Cardy}(2004)}]{calabrese_entanglement_2004}%
  \BibitemOpen
  \bibfield  {author} {\bibinfo {author} {\bibfnamefont {P.}~\bibnamefont
  {Calabrese}}\ and\ \bibinfo {author} {\bibfnamefont {J.}~\bibnamefont
  {Cardy}},\ }\href {https://doi.org/10.1088/1742-5468/2004/06/P06002}
  {\bibfield  {journal} {\bibinfo  {journal} {Journal of Statistical Mechanics:
  Theory and Experiment}\ }\textbf {\bibinfo {volume} {2004}},\ \bibinfo
  {pages} {P06002} (\bibinfo {year} {2004})}\BibitemShut {NoStop}%
\bibitem [{\citenamefont {Orus}\ \emph {et~al.}(2006)\citenamefont {Orus},
  \citenamefont {Latorre}, \citenamefont {Eisert},\ and\ \citenamefont
  {Cramer}}]{orus_half_2006}%
  \BibitemOpen
  \bibfield  {author} {\bibinfo {author} {\bibfnamefont {R.}~\bibnamefont
  {Orus}}, \bibinfo {author} {\bibfnamefont {J.~I.}\ \bibnamefont {Latorre}},
  \bibinfo {author} {\bibfnamefont {J.}~\bibnamefont {Eisert}},\ and\ \bibinfo
  {author} {\bibfnamefont {M.}~\bibnamefont {Cramer}},\ }\href
  {https://doi.org/10.1103/PhysRevA.73.060303} {\bibfield  {journal} {\bibinfo
  {journal} {Physical Review A}\ }\textbf {\bibinfo {volume} {73}},\ \bibinfo
  {pages} {060303} (\bibinfo {year} {2006})},\ \Eprint
  {https://arxiv.org/abs/quant-ph/0509023} {arXiv:quant-ph/0509023}
  \BibitemShut {NoStop}%
\bibitem [{\citenamefont {Wen}(1990)}]{wen_topological_1990}%
  \BibitemOpen
  \bibfield  {author} {\bibinfo {author} {\bibfnamefont {X.~G.}\ \bibnamefont
  {Wen}},\ }\href {https://doi.org/10.1142/S0217979290000139} {\bibfield
  {journal} {\bibinfo  {journal} {International Journal of Modern Physics B}\
  }\textbf {\bibinfo {volume} {04}},\ \bibinfo {pages} {239} (\bibinfo {year}
  {1990})}\BibitemShut {NoStop}%
\bibitem [{\citenamefont {Wen}(1995)}]{wen_topological_1995}%
  \BibitemOpen
  \bibfield  {author} {\bibinfo {author} {\bibfnamefont {X.-G.}\ \bibnamefont
  {Wen}},\ }\href {https://doi.org/10.1080/00018739500101566} {\bibfield
  {journal} {\bibinfo  {journal} {Advances in Physics}\ }\textbf {\bibinfo
  {volume} {44}},\ \bibinfo {pages} {405} (\bibinfo {year} {1995})}\BibitemShut
  {NoStop}%
\bibitem [{\citenamefont
  {Kitaev}(2003)}]{kitaevFaulttolerantQuantumComputation2003}%
  \BibitemOpen
  \bibfield  {author} {\bibinfo {author} {\bibfnamefont {A.~{\relax Yu}.}\
  \bibnamefont {Kitaev}},\ }\href
  {https://doi.org/10.1016/S0003-4916(02)00018-0} {\bibfield  {journal}
  {\bibinfo  {journal} {Annals of Physics}\ }\textbf {\bibinfo {volume}
  {303}},\ \bibinfo {pages} {2} (\bibinfo {year} {2003})}\BibitemShut {NoStop}%
\bibitem [{\citenamefont {Kitaev}(2006)}]{kitaevAnyonsExactlySolved2006}%
  \BibitemOpen
  \bibfield  {author} {\bibinfo {author} {\bibfnamefont {A.}~\bibnamefont
  {Kitaev}},\ }\href {https://doi.org/10.1016/j.aop.2005.10.005} {\bibfield
  {journal} {\bibinfo  {journal} {Annals of Physics}\ }\bibinfo {series}
  {January {{Special Issue}}},\ \textbf {\bibinfo {volume} {321}},\ \bibinfo
  {pages} {2} (\bibinfo {year} {2006})}\BibitemShut {NoStop}%
\bibitem [{\citenamefont {Amico}\ \emph {et~al.}(2008)\citenamefont {Amico},
  \citenamefont {Fazio}, \citenamefont {Osterloh},\ and\ \citenamefont
  {Vedral}}]{amicoEntanglementManybodySystems2008}%
  \BibitemOpen
  \bibfield  {author} {\bibinfo {author} {\bibfnamefont {L.}~\bibnamefont
  {Amico}}, \bibinfo {author} {\bibfnamefont {R.}~\bibnamefont {Fazio}},
  \bibinfo {author} {\bibfnamefont {A.}~\bibnamefont {Osterloh}},\ and\
  \bibinfo {author} {\bibfnamefont {V.}~\bibnamefont {Vedral}},\ }\href
  {https://doi.org/10.1103/RevModPhys.80.517} {\bibfield  {journal} {\bibinfo
  {journal} {Reviews of Modern Physics}\ }\textbf {\bibinfo {volume} {80}},\
  \bibinfo {pages} {517} (\bibinfo {year} {2008})}\BibitemShut {NoStop}%
\bibitem [{\citenamefont {Eisert}\ \emph {et~al.}(2010)\citenamefont {Eisert},
  \citenamefont {Cramer},\ and\ \citenamefont
  {Plenio}}]{eisertColloquiumAreaLaws2010}%
  \BibitemOpen
  \bibfield  {author} {\bibinfo {author} {\bibfnamefont {J.}~\bibnamefont
  {Eisert}}, \bibinfo {author} {\bibfnamefont {M.}~\bibnamefont {Cramer}},\
  and\ \bibinfo {author} {\bibfnamefont {M.~B.}\ \bibnamefont {Plenio}},\
  }\href {https://doi.org/10.1103/RevModPhys.82.277} {\bibfield  {journal}
  {\bibinfo  {journal} {Reviews of Modern Physics}\ }\textbf {\bibinfo {volume}
  {82}},\ \bibinfo {pages} {277} (\bibinfo {year} {2010})}\BibitemShut
  {NoStop}%
\bibitem [{\citenamefont {Hollands}\ and\ \citenamefont
  {Sanders}(2018)}]{hollands_entanglement_2018}%
  \BibitemOpen
  \bibfield  {author} {\bibinfo {author} {\bibfnamefont {S.}~\bibnamefont
  {Hollands}}\ and\ \bibinfo {author} {\bibfnamefont {K.}~\bibnamefont
  {Sanders}},\ }\href {https://doi.org/10.1007/978-3-319-94902-4} {\emph
  {\bibinfo {title} {Entanglement {Measures} and {Their} {Properties} in
  {Quantum} {Field} {Theory}}}},\ \bibinfo {series} {{SpringerBriefs} in
  {Mathematical} {Physics}}, Vol.~\bibinfo {volume} {34}\ (\bibinfo
  {publisher} {Springer International Publishing},\ \bibinfo {address} {Cham},\
  \bibinfo {year} {2018})\BibitemShut {NoStop}%
\bibitem [{\citenamefont {Summers}\ and\ \citenamefont
  {Werner}(1985)}]{summers_vacuum_1985}%
  \BibitemOpen
  \bibfield  {author} {\bibinfo {author} {\bibfnamefont {S.~J.}\ \bibnamefont
  {Summers}}\ and\ \bibinfo {author} {\bibfnamefont {R.}~\bibnamefont
  {Werner}},\ }\href {https://doi.org/10.1016/0375-9601(85)90093-3} {\bibfield
  {journal} {\bibinfo  {journal} {Physics Letters A}\ }\textbf {\bibinfo
  {volume} {110}},\ \bibinfo {pages} {257} (\bibinfo {year}
  {1985})}\BibitemShut {NoStop}%
\bibitem [{\citenamefont {Summers}\ and\ \citenamefont
  {Werner}(1987)}]{summers_maximal_1987}%
  \BibitemOpen
  \bibfield  {author} {\bibinfo {author} {\bibfnamefont {S.~J.}\ \bibnamefont
  {Summers}}\ and\ \bibinfo {author} {\bibfnamefont {R.}~\bibnamefont
  {Werner}},\ }\href {https://doi.org/10.1007/BF01207366} {\bibfield  {journal}
  {\bibinfo  {journal} {Communications in Mathematical Physics}\ }\textbf
  {\bibinfo {volume} {110}},\ \bibinfo {pages} {247} (\bibinfo {year}
  {1987})}\BibitemShut {NoStop}%
\bibitem [{\citenamefont {Verch}\ and\ \citenamefont
  {Werner}(2005)}]{verch_distillability_2005}%
  \BibitemOpen
  \bibfield  {author} {\bibinfo {author} {\bibfnamefont {R.}~\bibnamefont
  {Verch}}\ and\ \bibinfo {author} {\bibfnamefont {R.~F.}\ \bibnamefont
  {Werner}},\ }\href {https://doi.org/10.1142/S0129055X05002364} {\bibfield
  {journal} {\bibinfo  {journal} {Reviews in Mathematical Physics}\ }\textbf
  {\bibinfo {volume} {17}},\ \bibinfo {pages} {545} (\bibinfo {year}
  {2005})}\BibitemShut {NoStop}%
\bibitem [{\citenamefont {Matsui}(2001)}]{matsui_split_2001}%
  \BibitemOpen
  \bibfield  {author} {\bibinfo {author} {\bibfnamefont {T.}~\bibnamefont
  {Matsui}},\ }\href {https://doi.org/10.1007/s002200100413} {\bibfield
  {journal} {\bibinfo  {journal} {Communications in Mathematical Physics}\
  }\textbf {\bibinfo {volume} {218}},\ \bibinfo {pages} {393} (\bibinfo {year}
  {2001})}\BibitemShut {NoStop}%
\bibitem [{\citenamefont {Keyl}\ \emph {et~al.}(2006)\citenamefont {Keyl},
  \citenamefont {Matsui}, \citenamefont {Schlingemann},\ and\ \citenamefont
  {Werner}}]{keyl_entanglement_2006}%
  \BibitemOpen
  \bibfield  {author} {\bibinfo {author} {\bibfnamefont {M.}~\bibnamefont
  {Keyl}}, \bibinfo {author} {\bibfnamefont {T.}~\bibnamefont {Matsui}},
  \bibinfo {author} {\bibfnamefont {D.}~\bibnamefont {Schlingemann}},\ and\
  \bibinfo {author} {\bibfnamefont {R.~F.}\ \bibnamefont {Werner}},\ }\href
  {https://doi.org/10.1142/S0129055X0600284X} {\bibfield  {journal} {\bibinfo
  {journal} {Reviews in Mathematical Physics}\ }\textbf {\bibinfo {volume}
  {18}},\ \bibinfo {pages} {935} (\bibinfo {year} {2006})}\BibitemShut
  {NoStop}%
\bibitem [{\citenamefont {Matsui}(2013)}]{matsui_boundedness_2013}%
  \BibitemOpen
  \bibfield  {author} {\bibinfo {author} {\bibfnamefont {T.}~\bibnamefont
  {Matsui}},\ }\href {https://doi.org/10.1142/S0129055X13500177} {\bibfield
  {journal} {\bibinfo  {journal} {Reviews in Mathematical Physics}\ }\textbf
  {\bibinfo {volume} {25}},\ \bibinfo {pages} {1350017} (\bibinfo {year}
  {2013})},\ \bibinfo {note} {publisher: World Scientific Publishing
  Co.}\BibitemShut {Stop}%
\bibitem [{\citenamefont {van Luijk}\ \emph
  {et~al.}(2024{\natexlab{a}})\citenamefont {van Luijk}, \citenamefont
  {Stottmeister},\ and\ \citenamefont {Wilming}}]{van_luijk_critical_2024}%
  \BibitemOpen
  \bibfield  {author} {\bibinfo {author} {\bibfnamefont {L.}~\bibnamefont {van
  Luijk}}, \bibinfo {author} {\bibfnamefont {A.}~\bibnamefont {Stottmeister}},\
  and\ \bibinfo {author} {\bibfnamefont {H.}~\bibnamefont {Wilming}},\ }\href
  {http://arxiv.org/abs/2406.11747} {\bibinfo {title} {Critical fermions are
  universal embezzlers}} (\bibinfo {year} {2024}{\natexlab{a}}),\ \Eprint
  {https://arxiv.org/abs/2406.11747} {2406.11747} \BibitemShut {NoStop}%
\bibitem [{\citenamefont
  {Naaijkens}(2012)}]{naaijkensAnyonsInfiniteQuantum2012}%
  \BibitemOpen
  \bibfield  {author} {\bibinfo {author} {\bibfnamefont {P.}~\bibnamefont
  {Naaijkens}},\ }\emph {\bibinfo {title} {Anyons in Infinite Quantum Systems:
  {{QFT}} in D=2+1 and the {{Toric Code}}}},\ \href@noop {} {Ph.D. thesis},\
  \bibinfo  {school} {Radboud Universiteit Nijmegen} (\bibinfo {year}
  {2012})\BibitemShut {NoStop}%
\bibitem [{\citenamefont {Fiedler}\ and\ \citenamefont
  {Naaijkens}(2015)}]{fiedlerHaagDualityKitaevs2015}%
  \BibitemOpen
  \bibfield  {author} {\bibinfo {author} {\bibfnamefont {L.}~\bibnamefont
  {Fiedler}}\ and\ \bibinfo {author} {\bibfnamefont {P.}~\bibnamefont
  {Naaijkens}},\ }\href {https://doi.org/10.1142/S0129055X1550021X} {\bibfield
  {journal} {\bibinfo  {journal} {Reviews in Mathematical Physics}\ }\textbf
  {\bibinfo {volume} {27}},\ \bibinfo {pages} {1550021} (\bibinfo {year}
  {2015})}\BibitemShut {NoStop}%
\bibitem [{\citenamefont {Naaijkens}(2017)}]{naaijkensQuantumSpinSystems2017}%
  \BibitemOpen
  \bibfield  {author} {\bibinfo {author} {\bibfnamefont {P.}~\bibnamefont
  {Naaijkens}},\ }\href {https://doi.org/10.1007/978-3-319-51458-1} {\emph
  {\bibinfo {title} {Quantum {{Spin Systems}} on {{Infinite Lattices}}: {{A
  Concise Introduction}}}}},\ \bibinfo {series} {Lecture {{Notes}} in
  {{Physics}}}, Vol.\ \bibinfo {volume} {933}\ (\bibinfo  {publisher} {Springer
  International Publishing},\ \bibinfo {address} {Cham},\ \bibinfo {year}
  {2017})\BibitemShut {NoStop}%
\bibitem [{\citenamefont {Ogata}(2022{\natexlab{a}})}]{ogata_type_2022}%
  \BibitemOpen
  \bibfield  {author} {\bibinfo {author} {\bibfnamefont {Y.}~\bibnamefont
  {Ogata}},\ }\href {https://doi.org/10.48550/arXiv.2212.09036} {\bibinfo
  {title} {Type of local von neumann algebras in abelian quantum double model}}
  (\bibinfo {year} {2022}{\natexlab{a}}),\ \Eprint
  {https://arxiv.org/abs/2212.09036} {arXiv:2212.09036} \BibitemShut {NoStop}%
\bibitem [{\citenamefont {Naaijkens}\ and\ \citenamefont
  {Ogata}(2022)}]{naaijkensSplitApproximateSplit2022}%
  \BibitemOpen
  \bibfield  {author} {\bibinfo {author} {\bibfnamefont {P.}~\bibnamefont
  {Naaijkens}}\ and\ \bibinfo {author} {\bibfnamefont {Y.}~\bibnamefont
  {Ogata}},\ }\href {https://doi.org/10.1007/s00220-022-04356-3} {\bibfield
  {journal} {\bibinfo  {journal} {Communications in Mathematical Physics}\
  }\textbf {\bibinfo {volume} {392}},\ \bibinfo {pages} {921} (\bibinfo {year}
  {2022})}\BibitemShut {NoStop}%
\bibitem [{\citenamefont {Jones}\ \emph {et~al.}(2023)\citenamefont {Jones},
  \citenamefont {Naaijkens}, \citenamefont {Penneys},\ and\ \citenamefont
  {Wallick}}]{jonesLocalTopologicalOrder2023}%
  \BibitemOpen
  \bibfield  {author} {\bibinfo {author} {\bibfnamefont {C.}~\bibnamefont
  {Jones}}, \bibinfo {author} {\bibfnamefont {P.}~\bibnamefont {Naaijkens}},
  \bibinfo {author} {\bibfnamefont {D.}~\bibnamefont {Penneys}},\ and\ \bibinfo
  {author} {\bibfnamefont {D.}~\bibnamefont {Wallick}},\ }\href
  {https://doi.org/10.48550/arXiv.2307.12552} {\bibinfo {title} {Local
  topological order and boundary algebras}} (\bibinfo {year} {2023}),\ \Eprint
  {https://arxiv.org/abs/2307.12552} {arXiv:2307.12552 [cond-mat,
  physics:math-ph, physics:quant-ph]} \BibitemShut {NoStop}%
\bibitem [{\citenamefont {Tomba}\ \emph {et~al.}(2023)\citenamefont {Tomba},
  \citenamefont {Wei}, \citenamefont {Hungar}, \citenamefont {Wallick},
  \citenamefont {Kawagoe}, \citenamefont {Chuah},\ and\ \citenamefont
  {Penneys}}]{tombaBoundaryAlgebrasKitaev2023}%
  \BibitemOpen
  \bibfield  {author} {\bibinfo {author} {\bibfnamefont {M.}~\bibnamefont
  {Tomba}}, \bibinfo {author} {\bibfnamefont {S.}~\bibnamefont {Wei}}, \bibinfo
  {author} {\bibfnamefont {B.}~\bibnamefont {Hungar}}, \bibinfo {author}
  {\bibfnamefont {D.}~\bibnamefont {Wallick}}, \bibinfo {author} {\bibfnamefont
  {K.}~\bibnamefont {Kawagoe}}, \bibinfo {author} {\bibfnamefont {C.~Y.}\
  \bibnamefont {Chuah}},\ and\ \bibinfo {author} {\bibfnamefont
  {D.}~\bibnamefont {Penneys}},\ }\href
  {https://doi.org/10.48550/arXiv.2309.13440} {\bibinfo {title} {Boundary
  algebras of the {{Kitaev Quantum Double}} model}} (\bibinfo {year} {2023}),\
  \Eprint {https://arxiv.org/abs/2309.13440} {arXiv:2309.13440} \BibitemShut
  {NoStop}%
\bibitem [{\citenamefont {Bhardwaj}\ \emph {et~al.}(2024)\citenamefont
  {Bhardwaj}, \citenamefont {Brisky}, \citenamefont {Chuah}, \citenamefont
  {Kawagoe}, \citenamefont {Keslin}, \citenamefont {Penneys},\ and\
  \citenamefont {Wallick}}]{bhardwaj_superselection_2024}%
  \BibitemOpen
  \bibfield  {author} {\bibinfo {author} {\bibfnamefont {A.}~\bibnamefont
  {Bhardwaj}}, \bibinfo {author} {\bibfnamefont {T.}~\bibnamefont {Brisky}},
  \bibinfo {author} {\bibfnamefont {C.~Y.}\ \bibnamefont {Chuah}}, \bibinfo
  {author} {\bibfnamefont {K.}~\bibnamefont {Kawagoe}}, \bibinfo {author}
  {\bibfnamefont {J.}~\bibnamefont {Keslin}}, \bibinfo {author} {\bibfnamefont
  {D.}~\bibnamefont {Penneys}},\ and\ \bibinfo {author} {\bibfnamefont
  {D.}~\bibnamefont {Wallick}},\ }\href
  {https://doi.org/10.48550/arXiv.2410.21454} {\bibinfo {title} {Superselection
  sectors for posets of von {Neumann} algebras}} (\bibinfo {year} {2024}),\
  \Eprint {https://arxiv.org/abs/2410.21454} {arXiv:2410.21454} \BibitemShut
  {NoStop}%
\bibitem [{\citenamefont {Bratteli}\ and\ \citenamefont
  {Robinson}(1987)}]{BR1}%
  \BibitemOpen
  \bibfield  {author} {\bibinfo {author} {\bibfnamefont {O.}~\bibnamefont
  {Bratteli}}\ and\ \bibinfo {author} {\bibfnamefont {D.~W.}\ \bibnamefont
  {Robinson}},\ }\href {https://doi.org/10.1007/978-3-662-02520-8} {\emph
  {\bibinfo {title} {Operator {{Algebras}} and {{Quantum Statistical
  Mechanics}} 1}}}\ (\bibinfo  {publisher} {Springer},\ \bibinfo {address}
  {Berlin, Heidelberg},\ \bibinfo {year} {1987})\BibitemShut {NoStop}%
\bibitem [{\citenamefont {Bratteli}\ and\ \citenamefont
  {Robinson}(1997)}]{BR2}%
  \BibitemOpen
  \bibfield  {author} {\bibinfo {author} {\bibfnamefont {O.}~\bibnamefont
  {Bratteli}}\ and\ \bibinfo {author} {\bibfnamefont {D.~W.}\ \bibnamefont
  {Robinson}},\ }\href {https://doi.org/10.1007/978-3-662-03444-6} {\emph
  {\bibinfo {title} {Operator {{Algebras}} and {{Quantum Statistical
  Mechanics}} 2}}}\ (\bibinfo  {publisher} {Springer},\ \bibinfo {address}
  {Berlin, Heidelberg},\ \bibinfo {year} {1997})\BibitemShut {NoStop}%
\bibitem [{\citenamefont {van Luijk}\ \emph
  {et~al.}(2024{\natexlab{b}})\citenamefont {van Luijk}, \citenamefont
  {Stottmeister}, \citenamefont {Werner},\ and\ \citenamefont
  {Wilming}}]{van_luijk_relativistic_2024}%
  \BibitemOpen
  \bibfield  {author} {\bibinfo {author} {\bibfnamefont {L.}~\bibnamefont {van
  Luijk}}, \bibinfo {author} {\bibfnamefont {A.}~\bibnamefont {Stottmeister}},
  \bibinfo {author} {\bibfnamefont {R.~F.}\ \bibnamefont {Werner}},\ and\
  \bibinfo {author} {\bibfnamefont {H.}~\bibnamefont {Wilming}},\ }\href
  {https://doi.org/10.1103/PhysRevLett.133.261602} {\bibfield  {journal}
  {\bibinfo  {journal} {Physical Review Letters}\ }\textbf {\bibinfo {volume}
  {133}},\ \bibinfo {pages} {261602} (\bibinfo {year}
  {2024}{\natexlab{b}})}\BibitemShut {NoStop}%
\bibitem [{\citenamefont {van Luijk}\ \emph
  {et~al.}(2024{\natexlab{c}})\citenamefont {van Luijk}, \citenamefont
  {Stottmeister}, \citenamefont {Werner},\ and\ \citenamefont
  {Wilming}}]{long_paper}%
  \BibitemOpen
  \bibfield  {author} {\bibinfo {author} {\bibfnamefont {L.}~\bibnamefont {van
  Luijk}}, \bibinfo {author} {\bibfnamefont {A.}~\bibnamefont {Stottmeister}},
  \bibinfo {author} {\bibfnamefont {R.~F.}\ \bibnamefont {Werner}},\ and\
  \bibinfo {author} {\bibfnamefont {H.}~\bibnamefont {Wilming}},\ }\href
  {https://doi.org/10.48550/arXiv.2401.07299} {\bibinfo {title} {Embezzlement
  of entanglement, quantum fields, and the classification of von {Neumann}
  algebras}} (\bibinfo {year} {2024}{\natexlab{c}}),\ \Eprint
  {https://arxiv.org/abs/2401.07299} {arXiv:2401.07299} \BibitemShut {NoStop}%
\bibitem [{\citenamefont {Keyl}\ \emph {et~al.}(2003)\citenamefont {Keyl},
  \citenamefont {Schlingemann},\ and\ \citenamefont
  {Werner}}]{keyl_infinitely_2003}%
  \BibitemOpen
  \bibfield  {author} {\bibinfo {author} {\bibfnamefont {M.}~\bibnamefont
  {Keyl}}, \bibinfo {author} {\bibfnamefont {D.}~\bibnamefont {Schlingemann}},\
  and\ \bibinfo {author} {\bibfnamefont {R.}~\bibnamefont {Werner}},\ }\href
  {https://doi.org/10.26421/QIC3.4-1} {\bibfield  {journal} {\bibinfo
  {journal} {Quantum Information \& Computation}\ }\textbf {\bibinfo {volume}
  {3}},\ \bibinfo {pages} {281} (\bibinfo {year} {2003})}\BibitemShut {NoStop}%
\bibitem [{\citenamefont {van Dam}\ and\ \citenamefont
  {Hayden}(2003)}]{van_dam_universal_2003}%
  \BibitemOpen
  \bibfield  {author} {\bibinfo {author} {\bibfnamefont {W.}~\bibnamefont {van
  Dam}}\ and\ \bibinfo {author} {\bibfnamefont {P.}~\bibnamefont {Hayden}},\
  }\href {https://doi.org/10.1103/PhysRevA.67.060302} {\bibfield  {journal}
  {\bibinfo  {journal} {Physical Review A}\ }\textbf {\bibinfo {volume} {67}},\
  \bibinfo {pages} {060302} (\bibinfo {year} {2003})}\BibitemShut {NoStop}%
\bibitem [{\citenamefont {Evans}\ and\ \citenamefont
  {Kawahigashi}(1998)}]{evans1998qsym}%
  \BibitemOpen
  \bibfield  {author} {\bibinfo {author} {\bibfnamefont {D.~E.}\ \bibnamefont
  {Evans}}\ and\ \bibinfo {author} {\bibfnamefont {Y.}~\bibnamefont
  {Kawahigashi}},\ }\href@noop {} {\emph {\bibinfo {title} {Quantum symmetries
  on operator algebras}}},\ Oxford Mathematical Monographs\ (\bibinfo
  {publisher} {The Clarendon Press, Oxford University Press, New York},\
  \bibinfo {year} {1998})\ pp.\ \bibinfo {pages} {xvi+829},\ \bibinfo {note}
  {oxford Science Publications}\BibitemShut {NoStop}%
\bibitem [{\citenamefont {Connes}(1976)}]{connes_classification_1976}%
  \BibitemOpen
  \bibfield  {author} {\bibinfo {author} {\bibfnamefont {A.}~\bibnamefont
  {Connes}},\ }\href {https://doi.org/10.2307/1971057} {\bibfield  {journal}
  {\bibinfo  {journal} {Annals of Mathematics}\ }\textbf {\bibinfo {volume}
  {104}},\ \bibinfo {pages} {73} (\bibinfo {year} {1976})}\BibitemShut
  {NoStop}%
\bibitem [{\citenamefont {Takesaki}(2003{\natexlab{a}})}]{takesaki3}%
  \BibitemOpen
  \bibfield  {author} {\bibinfo {author} {\bibfnamefont {M.}~\bibnamefont
  {Takesaki}},\ }\href {https://doi.org/10.1007/978-3-662-10453-8} {\emph
  {\bibinfo {title} {Theory of {Operator} {Algebras} {III}}}},\ edited by\
  \bibinfo {editor} {\bibfnamefont {J.}~\bibnamefont {Cuntz}}\ and\ \bibinfo
  {editor} {\bibfnamefont {V.~F.~R.}\ \bibnamefont {Jones}},\ \bibinfo {series}
  {Encyclopaedia of {Mathematical} {Sciences}}, Vol.\ \bibinfo {volume} {127}\
  (\bibinfo  {publisher} {Springer},\ \bibinfo {address} {Berlin, Heidelberg},\
  \bibinfo {year} {2003})\BibitemShut {NoStop}%
\bibitem [{\citenamefont {van Luijk}\ \emph
  {et~al.}(2024{\natexlab{d}})\citenamefont {van Luijk}, \citenamefont
  {Stottmeister}, \citenamefont {Werner},\ and\ \citenamefont
  {Wilming}}]{van_luijk_pure_2024}%
  \BibitemOpen
  \bibfield  {author} {\bibinfo {author} {\bibfnamefont {L.}~\bibnamefont {van
  Luijk}}, \bibinfo {author} {\bibfnamefont {A.}~\bibnamefont {Stottmeister}},
  \bibinfo {author} {\bibfnamefont {R.~F.}\ \bibnamefont {Werner}},\ and\
  \bibinfo {author} {\bibfnamefont {H.}~\bibnamefont {Wilming}},\ }\href
  {https://doi.org/10.48550/arXiv.2409.17739} {\bibinfo {title} {Pure state
  entanglement and von neumann algebras}} (\bibinfo {year}
  {2024}{\natexlab{d}}),\ \Eprint {https://arxiv.org/abs/2409.17739}
  {arXiv:2409.17739} \BibitemShut {NoStop}%
\bibitem [{\citenamefont {Connes}(1973)}]{connes1973classIII}%
  \BibitemOpen
  \bibfield  {author} {\bibinfo {author} {\bibfnamefont {A.}~\bibnamefont
  {Connes}},\ }\href {https://doi.org/10.24033/asens.1247} {\bibfield
  {journal} {\bibinfo  {journal} {Annales scientifiques de l'École normale
  supérieure}\ }\textbf {\bibinfo {volume} {6}},\ \bibinfo {pages} {133}
  (\bibinfo {year} {1973})}\BibitemShut {NoStop}%
\bibitem [{\citenamefont {Haagerup}(1987)}]{haagerup_connes_1987}%
  \BibitemOpen
  \bibfield  {author} {\bibinfo {author} {\bibfnamefont {U.}~\bibnamefont
  {Haagerup}},\ }\href {https://doi.org/10.1007/BF02392257} {\bibfield
  {journal} {\bibinfo  {journal} {Acta Mathematica}\ }\textbf {\bibinfo
  {volume} {158}},\ \bibinfo {pages} {95} (\bibinfo {year} {1987})},\ \bibinfo
  {note} {publisher: Institut Mittag-Leffler}\BibitemShut {NoStop}%
\bibitem [{\citenamefont {Takesaki}(2003{\natexlab{b}})}]{takesaki2}%
  \BibitemOpen
  \bibfield  {author} {\bibinfo {author} {\bibfnamefont {M.}~\bibnamefont
  {Takesaki}},\ }\href {https://doi.org/10.1007/978-3-662-10451-4} {\emph
  {\bibinfo {title} {Theory of {Operator} {Algebras} {II}}}},\ edited by\
  \bibinfo {editor} {\bibfnamefont {J.}~\bibnamefont {Cuntz}}\ and\ \bibinfo
  {editor} {\bibfnamefont {V.~F.~R.}\ \bibnamefont {Jones}},\ \bibinfo {series}
  {Encyclopaedia of {Mathematical} {Sciences}}, Vol.\ \bibinfo {volume} {125}\
  (\bibinfo  {publisher} {Springer},\ \bibinfo {address} {Berlin, Heidelberg},\
  \bibinfo {year} {2003})\BibitemShut {NoStop}%
\bibitem [{\citenamefont {Crann}\ \emph {et~al.}(2020)\citenamefont {Crann},
  \citenamefont {Kribs}, \citenamefont {Levene},\ and\ \citenamefont
  {Todorov}}]{crann_state_2020}%
  \BibitemOpen
  \bibfield  {author} {\bibinfo {author} {\bibfnamefont {J.}~\bibnamefont
  {Crann}}, \bibinfo {author} {\bibfnamefont {D.~W.}\ \bibnamefont {Kribs}},
  \bibinfo {author} {\bibfnamefont {R.~H.}\ \bibnamefont {Levene}},\ and\
  \bibinfo {author} {\bibfnamefont {I.~G.}\ \bibnamefont {Todorov}},\ }\href
  {https://doi.org/10.1007/s00220-020-03803-3} {\bibfield  {journal} {\bibinfo
  {journal} {Communications in Mathematical Physics}\ }\textbf {\bibinfo
  {volume} {378}},\ \bibinfo {pages} {1123} (\bibinfo {year}
  {2020})}\BibitemShut {NoStop}%
\bibitem [{\citenamefont {Ogata}(2022{\natexlab{b}})}]{ogata_derivation_2022}%
  \BibitemOpen
  \bibfield  {author} {\bibinfo {author} {\bibfnamefont {Y.}~\bibnamefont
  {Ogata}},\ }\href {https://doi.org/10.1063/5.0061785} {\bibfield  {journal}
  {\bibinfo  {journal} {Journal of Mathematical Physics}\ }\textbf {\bibinfo
  {volume} {63}},\ \bibinfo {pages} {011902} (\bibinfo {year}
  {2022}{\natexlab{b}})},\ \Eprint {https://arxiv.org/abs/2106.15741}
  {2106.15741} \BibitemShut {NoStop}%
\bibitem [{\citenamefont
  {Naaijkens}(2011)}]{naaijkensLocalizedEndomorphismsKitaevs2011}%
  \BibitemOpen
  \bibfield  {author} {\bibinfo {author} {\bibfnamefont {P.}~\bibnamefont
  {Naaijkens}},\ }\href {https://doi.org/10.1142/S0129055X1100431X} {\bibfield
  {journal} {\bibinfo  {journal} {Reviews in Mathematical Physics}\ }\textbf
  {\bibinfo {volume} {23}},\ \bibinfo {pages} {347} (\bibinfo {year} {2011})},\
  \Eprint {https://arxiv.org/abs/1012.3857} {arXiv:1012.3857 [math-ph,
  physics:quant-ph]} \BibitemShut {NoStop}%
\bibitem [{\citenamefont {Fiedler}\ \emph {et~al.}(2017)\citenamefont
  {Fiedler}, \citenamefont {Naaijkens},\ and\ \citenamefont
  {Osborne}}]{fiedlerJonesIndexSecret2017}%
  \BibitemOpen
  \bibfield  {author} {\bibinfo {author} {\bibfnamefont {L.}~\bibnamefont
  {Fiedler}}, \bibinfo {author} {\bibfnamefont {P.}~\bibnamefont {Naaijkens}},\
  and\ \bibinfo {author} {\bibfnamefont {T.~J.}\ \bibnamefont {Osborne}},\
  }\href {https://doi.org/10.1088/1367-2630/aa5c0c} {\bibfield  {journal}
  {\bibinfo  {journal} {New Journal of Physics}\ }\textbf {\bibinfo {volume}
  {19}},\ \bibinfo {pages} {023039} (\bibinfo {year} {2017})}\BibitemShut
  {NoStop}%
\bibitem [{\citenamefont {Gao}\ \emph {et~al.}(2020)\citenamefont {Gao},
  \citenamefont {Junge},\ and\ \citenamefont
  {LaRacuente}}]{gaoRelativeEntropyNeumann2020}%
  \BibitemOpen
  \bibfield  {author} {\bibinfo {author} {\bibfnamefont {L.}~\bibnamefont
  {Gao}}, \bibinfo {author} {\bibfnamefont {M.}~\bibnamefont {Junge}},\ and\
  \bibinfo {author} {\bibfnamefont {N.}~\bibnamefont {LaRacuente}},\ }\href
  {https://doi.org/10.1142/S0129167X20500469} {\bibfield  {journal} {\bibinfo
  {journal} {International Journal of Mathematics}\ }\textbf {\bibinfo {volume}
  {31}},\ \bibinfo {pages} {2050046} (\bibinfo {year} {2020})}\BibitemShut
  {NoStop}%
\bibitem [{\citenamefont {van Luijk}\ \emph
  {et~al.}(2024{\natexlab{e}})\citenamefont {van Luijk}, \citenamefont
  {Stottmeister},\ and\ \citenamefont {Wilming}}]{van_luijk_multipartite_2024}%
  \BibitemOpen
  \bibfield  {author} {\bibinfo {author} {\bibfnamefont {L.}~\bibnamefont {van
  Luijk}}, \bibinfo {author} {\bibfnamefont {A.}~\bibnamefont {Stottmeister}},\
  and\ \bibinfo {author} {\bibfnamefont {H.}~\bibnamefont {Wilming}},\ }\href
  {https://arxiv.org/abs/2409.07646v1} {{\selectlanguage {english}\bibinfo
  {title} {Multipartite embezzlement of entanglement}}} (\bibinfo {year}
  {2024}{\natexlab{e}}),\ \Eprint {https://arxiv.org/abs/2409.07646}
  {arXiv:2409.07646} \BibitemShut {NoStop}%
\bibitem [{\citenamefont {Hastings}(2007)}]{hastings_area_2007}%
  \BibitemOpen
  \bibfield  {author} {\bibinfo {author} {\bibfnamefont {M.~B.}\ \bibnamefont
  {Hastings}},\ }\href {https://doi.org/10.1088/1742-5468/2007/08/P08024}
  {\bibfield  {journal} {\bibinfo  {journal} {Journal of Statistical Mechanics:
  Theory and Experiment}\ }\textbf {\bibinfo {volume} {2007}},\ \bibinfo
  {pages} {P08024} (\bibinfo {year} {2007})}\BibitemShut {NoStop}%
\bibitem [{\citenamefont {Ukai}(2024)}]{ukai_hastings_2024}%
  \BibitemOpen
  \bibfield  {author} {\bibinfo {author} {\bibfnamefont {A.}~\bibnamefont
  {Ukai}},\ }\href {https://doi.org/10.48550/arXiv.2407.12324} {\bibinfo
  {title} {On hastings factorization for quantum many-body systems in the
  infinite volume setting}} (\bibinfo {year} {2024}),\ \Eprint
  {https://arxiv.org/abs/2407.12324} {arXiv:2407.12324} \BibitemShut {NoStop}%
\bibitem [{\citenamefont {Werner}(2024)}]{ReinhardPrivate}%
  \BibitemOpen
  \bibfield  {author} {\bibinfo {author} {\bibfnamefont {R.~F.}\ \bibnamefont
  {Werner}},\ }\href@noop {} {\bibfield  {journal} {\bibinfo  {journal}
  {Private communication}\ } (\bibinfo {year} {2024})}\BibitemShut {NoStop}%
\bibitem [{\citenamefont {Fannes}\ \emph {et~al.}(1992)\citenamefont {Fannes},
  \citenamefont {Nachtergaele},\ and\ \citenamefont
  {Werner}}]{fannesFinitelyCorrelatedStates1992}%
  \BibitemOpen
  \bibfield  {author} {\bibinfo {author} {\bibfnamefont {M.}~\bibnamefont
  {Fannes}}, \bibinfo {author} {\bibfnamefont {B.}~\bibnamefont
  {Nachtergaele}},\ and\ \bibinfo {author} {\bibfnamefont {R.~F.}\ \bibnamefont
  {Werner}},\ }\href {https://doi.org/10.1007/BF02099178} {\bibfield  {journal}
  {\bibinfo  {journal} {Communications in Mathematical Physics}\ }\textbf
  {\bibinfo {volume} {144}},\ \bibinfo {pages} {443} (\bibinfo {year}
  {1992})}\BibitemShut {NoStop}%
\bibitem [{\citenamefont {Calabrese}\ and\ \citenamefont
  {Lefevre}(2008)}]{calabrese_entanglement_2008}%
  \BibitemOpen
  \bibfield  {author} {\bibinfo {author} {\bibfnamefont {P.}~\bibnamefont
  {Calabrese}}\ and\ \bibinfo {author} {\bibfnamefont {A.}~\bibnamefont
  {Lefevre}},\ }\href {https://doi.org/10.1103/PhysRevA.78.032329} {\bibfield
  {journal} {\bibinfo  {journal} {Physical Review A}\ }\textbf {\bibinfo
  {volume} {78}},\ \bibinfo {pages} {032329} (\bibinfo {year}
  {2008})}\BibitemShut {NoStop}%
\bibitem [{\citenamefont {Movassagh}\ and\ \citenamefont
  {Shor}(2016)}]{movassaghSupercriticalEntanglementLocal2016}%
  \BibitemOpen
  \bibfield  {author} {\bibinfo {author} {\bibfnamefont {R.}~\bibnamefont
  {Movassagh}}\ and\ \bibinfo {author} {\bibfnamefont {P.~W.}\ \bibnamefont
  {Shor}},\ }\href {https://doi.org/10.1073/pnas.1605716113} {\bibfield
  {journal} {\bibinfo  {journal} {Proceedings of the National Academy of
  Sciences}\ }\textbf {\bibinfo {volume} {113}},\ \bibinfo {pages} {13278}
  (\bibinfo {year} {2016})}\BibitemShut {NoStop}%
\bibitem [{\citenamefont {Salberger}\ and\ \citenamefont
  {Korepin}(2017)}]{salbergerFredkinSpinChain2017}%
  \BibitemOpen
  \bibfield  {author} {\bibinfo {author} {\bibfnamefont {O.}~\bibnamefont
  {Salberger}}\ and\ \bibinfo {author} {\bibfnamefont {V.}~\bibnamefont
  {Korepin}},\ }in\ \href {https://doi.org/10.1142/9789813233867_0022} {\emph
  {\bibinfo {booktitle} {Ludwig {{Faddeev Memorial Volume}}}}}\ (\bibinfo
  {publisher} {WORLD SCIENTIFIC},\ \bibinfo {year} {2017})\ pp.\ \bibinfo
  {pages} {439--458}\BibitemShut {NoStop}%
\bibitem [{\citenamefont {Salberger}\ \emph {et~al.}(2017)\citenamefont
  {Salberger}, \citenamefont {Udagawa}, \citenamefont {Zhang}, \citenamefont
  {Katsura}, \citenamefont {Klich},\ and\ \citenamefont
  {Korepin}}]{salbergerDeformedFredkinSpin2017}%
  \BibitemOpen
  \bibfield  {author} {\bibinfo {author} {\bibfnamefont {O.}~\bibnamefont
  {Salberger}}, \bibinfo {author} {\bibfnamefont {T.}~\bibnamefont {Udagawa}},
  \bibinfo {author} {\bibfnamefont {Z.}~\bibnamefont {Zhang}}, \bibinfo
  {author} {\bibfnamefont {H.}~\bibnamefont {Katsura}}, \bibinfo {author}
  {\bibfnamefont {I.}~\bibnamefont {Klich}},\ and\ \bibinfo {author}
  {\bibfnamefont {V.}~\bibnamefont {Korepin}},\ }\href
  {https://doi.org/10.1088/1742-5468/aa6b1f} {\bibfield  {journal} {\bibinfo
  {journal} {Journal of Statistical Mechanics: Theory and Experiment}\ }\textbf
  {\bibinfo {volume} {2017}},\ \bibinfo {pages} {063103} (\bibinfo {year}
  {2017})}\BibitemShut {NoStop}%
\bibitem [{\citenamefont {Zhang}\ \emph {et~al.}(2017)\citenamefont {Zhang},
  \citenamefont {Ahmadain},\ and\ \citenamefont
  {Klich}}]{zhangQuantumPhaseTransition2016a}%
  \BibitemOpen
  \bibfield  {author} {\bibinfo {author} {\bibfnamefont {Z.}~\bibnamefont
  {Zhang}}, \bibinfo {author} {\bibfnamefont {A.}~\bibnamefont {Ahmadain}},\
  and\ \bibinfo {author} {\bibfnamefont {I.}~\bibnamefont {Klich}},\ }\href
  {https://doi.org/10.1073/pnas.1702029114} {\bibfield  {journal} {\bibinfo
  {journal} {Proceedings of the National Academy of Sciences}\ }\textbf
  {\bibinfo {volume} {114}},\ \bibinfo {pages} {5142} (\bibinfo {year}
  {2017})},\ \Eprint {https://arxiv.org/abs/1606.07795} {arXiv:1606.07795}
  \BibitemShut {NoStop}%
\bibitem [{\citenamefont {Menon}\ \emph {et~al.}(2024)\citenamefont {Menon},
  \citenamefont {Gu},\ and\ \citenamefont
  {Movassagh}}]{menonSymmetriesCorrelationFunctions2024}%
  \BibitemOpen
  \bibfield  {author} {\bibinfo {author} {\bibfnamefont {V.}~\bibnamefont
  {Menon}}, \bibinfo {author} {\bibfnamefont {A.}~\bibnamefont {Gu}},\ and\
  \bibinfo {author} {\bibfnamefont {R.}~\bibnamefont {Movassagh}},\ }\href
  {https://doi.org/10.48550/arXiv.2408.16070} {\bibinfo {title} {Symmetries,
  correlation functions, and entanglement of general quantum {{Motzkin}}
  spin-chains}} (\bibinfo {year} {2024}),\ \Eprint
  {https://arxiv.org/abs/2408.16070} {arXiv:2408.16070 [cond-mat,
  physics:math-ph, physics:quant-ph]} \BibitemShut {NoStop}%
\bibitem [{\citenamefont {Keyl}\ \emph {et~al.}(2008)\citenamefont {Keyl},
  \citenamefont {Matsui}, \citenamefont {Schlingemann},\ and\ \citenamefont
  {Werner}}]{keyl2008haag}%
  \BibitemOpen
  \bibfield  {author} {\bibinfo {author} {\bibfnamefont {M.}~\bibnamefont
  {Keyl}}, \bibinfo {author} {\bibfnamefont {T.}~\bibnamefont {Matsui}},
  \bibinfo {author} {\bibfnamefont {D.}~\bibnamefont {Schlingemann}},\ and\
  \bibinfo {author} {\bibfnamefont {R.~F.}\ \bibnamefont {Werner}},\ }\href
  {https://doi.org/10.1142/S0129055X08003377} {\bibfield  {journal} {\bibinfo
  {journal} {Reviews in Mathematical Physics}\ }\textbf {\bibinfo {volume}
  {20}},\ \bibinfo {pages} {707} (\bibinfo {year} {2008})}\BibitemShut
  {NoStop}%
\bibitem [{\citenamefont {van Luijk}\ \emph
  {et~al.}(2024{\natexlab{f}})\citenamefont {van Luijk}, \citenamefont
  {Schwonnek}, \citenamefont {Stottmeister},\ and\ \citenamefont
  {Werner}}]{van_luijk_schmidt_2024}%
  \BibitemOpen
  \bibfield  {author} {\bibinfo {author} {\bibfnamefont {L.}~\bibnamefont {van
  Luijk}}, \bibinfo {author} {\bibfnamefont {R.}~\bibnamefont {Schwonnek}},
  \bibinfo {author} {\bibfnamefont {A.}~\bibnamefont {Stottmeister}},\ and\
  \bibinfo {author} {\bibfnamefont {R.~F.}\ \bibnamefont {Werner}},\ }\href
  {https://doi.org/10.1007/s00220-024-05011-9} {\bibfield  {journal} {\bibinfo
  {journal} {Communications in Mathematical Physics}\ }\textbf {\bibinfo
  {volume} {405}},\ \bibinfo {pages} {152} (\bibinfo {year}
  {2024}{\natexlab{f}})}\BibitemShut {NoStop}%
\bibitem [{\citenamefont {{Perez-Garcia}}\ \emph
  {et~al.}(2007{\natexlab{a}})\citenamefont {{Perez-Garcia}}, \citenamefont
  {Verstraete}, \citenamefont {Wolf},\ and\ \citenamefont
  {Cirac}}]{perez-garciaMatrixProductState2007}%
  \BibitemOpen
  \bibfield  {author} {\bibinfo {author} {\bibfnamefont {D.}~\bibnamefont
  {{Perez-Garcia}}}, \bibinfo {author} {\bibfnamefont {F.}~\bibnamefont
  {Verstraete}}, \bibinfo {author} {\bibfnamefont {M.~M.}\ \bibnamefont
  {Wolf}},\ and\ \bibinfo {author} {\bibfnamefont {J.~I.}\ \bibnamefont
  {Cirac}},\ }\href@noop {} {\bibfield  {journal} {\bibinfo  {journal} {Quantum
  Inf. Comput.}\ }\textbf {\bibinfo {volume} {7}} (\bibinfo {year}
  {2007}{\natexlab{a}})},\ \Eprint {https://arxiv.org/abs/quant-ph/0608197}
  {arXiv:quant-ph/0608197} \BibitemShut {NoStop}%
\bibitem [{\citenamefont {Levin}\ and\ \citenamefont
  {Wen}(2005)}]{levinStringnetCondensationPhysical2005}%
  \BibitemOpen
  \bibfield  {author} {\bibinfo {author} {\bibfnamefont {M.~A.}\ \bibnamefont
  {Levin}}\ and\ \bibinfo {author} {\bibfnamefont {X.-G.}\ \bibnamefont
  {Wen}},\ }\href {https://doi.org/10.1103/PhysRevB.71.045110} {\bibfield
  {journal} {\bibinfo  {journal} {Physical Review B}\ }\textbf {\bibinfo
  {volume} {71}},\ \bibinfo {pages} {045110} (\bibinfo {year}
  {2005})}\BibitemShut {NoStop}%
\bibitem [{\citenamefont {Kirillov~Jr.}(2011)}]{kirillov_jr_string-net_2011}%
  \BibitemOpen
  \bibfield  {author} {\bibinfo {author} {\bibfnamefont {A.}~\bibnamefont
  {Kirillov~Jr.}},\ }\href {https://doi.org/10.48550/arXiv.1106.6033} {\bibinfo
  {title} {String-net model of {Turaev}-{Viro} invariants}} (\bibinfo {year}
  {2011}),\ \bibinfo {note} {arXiv:1106.6033}\BibitemShut {NoStop}%
\bibitem [{\citenamefont {Kitaev}\ and\ \citenamefont
  {Kong}(2012)}]{kitaevModelsGappedBoundaries2012}%
  \BibitemOpen
  \bibfield  {author} {\bibinfo {author} {\bibfnamefont {A.}~\bibnamefont
  {Kitaev}}\ and\ \bibinfo {author} {\bibfnamefont {L.}~\bibnamefont {Kong}},\
  }\href {https://doi.org/10.1007/s00220-012-1500-5} {\bibfield  {journal}
  {\bibinfo  {journal} {Communications in Mathematical Physics}\ }\textbf
  {\bibinfo {volume} {313}},\ \bibinfo {pages} {351} (\bibinfo {year}
  {2012})}\BibitemShut {NoStop}%
\bibitem [{\citenamefont {Bultinck}\ \emph {et~al.}(2017)\citenamefont
  {Bultinck}, \citenamefont {Mariën}, \citenamefont {Williamson},
  \citenamefont {Şahinoğlu}, \citenamefont {Haegeman},\ and\ \citenamefont
  {Verstraete}}]{bultinck2017anyons}%
  \BibitemOpen
  \bibfield  {author} {\bibinfo {author} {\bibfnamefont {N.}~\bibnamefont
  {Bultinck}}, \bibinfo {author} {\bibfnamefont {M.}~\bibnamefont {Mariën}},
  \bibinfo {author} {\bibfnamefont {D.~J.}\ \bibnamefont {Williamson}},
  \bibinfo {author} {\bibfnamefont {M.~B.}\ \bibnamefont {Şahinoğlu}},
  \bibinfo {author} {\bibfnamefont {J.}~\bibnamefont {Haegeman}},\ and\
  \bibinfo {author} {\bibfnamefont {F.}~\bibnamefont {Verstraete}},\ }\href
  {https://doi.org/10.1016/j.aop.2017.01.004} {\bibfield  {journal} {\bibinfo
  {journal} {Annals of Physics}\ }\textbf {\bibinfo {volume} {378}},\ \bibinfo
  {pages} {183} (\bibinfo {year} {2017})},\ \bibinfo {note} {publisher:
  Elsevier Inc.}\BibitemShut {Stop}%
\bibitem [{\citenamefont {Green}\ \emph {et~al.}(2024)\citenamefont {Green},
  \citenamefont {Huston}, \citenamefont {Kawagoe}, \citenamefont {Penneys},
  \citenamefont {Poudel},\ and\ \citenamefont {Sanford}}]{green_enriched_2024}%
  \BibitemOpen
  \bibfield  {author} {\bibinfo {author} {\bibfnamefont {D.}~\bibnamefont
  {Green}}, \bibinfo {author} {\bibfnamefont {P.}~\bibnamefont {Huston}},
  \bibinfo {author} {\bibfnamefont {K.}~\bibnamefont {Kawagoe}}, \bibinfo
  {author} {\bibfnamefont {D.}~\bibnamefont {Penneys}}, \bibinfo {author}
  {\bibfnamefont {A.}~\bibnamefont {Poudel}},\ and\ \bibinfo {author}
  {\bibfnamefont {S.}~\bibnamefont {Sanford}},\ }\href
  {https://doi.org/10.22331/q-2024-03-28-1301} {\bibfield  {journal} {\bibinfo
  {journal} {Quantum}\ }\textbf {\bibinfo {volume} {8}},\ \bibinfo {pages}
  {1301} (\bibinfo {year} {2024})},\ \bibinfo {note} {arXiv:2305.14068
  [cond-mat]}\BibitemShut {NoStop}%
\bibitem [{\citenamefont {Müger}(2003)}]{muger_subfactors_2003}%
  \BibitemOpen
  \bibfield  {author} {\bibinfo {author} {\bibfnamefont {M.}~\bibnamefont
  {Müger}},\ }\href {https://doi.org/10.1016/S0022-4049(02)00247-5} {\bibfield
   {journal} {\bibinfo  {journal} {Journal of Pure and Applied Algebra}\
  }\textbf {\bibinfo {volume} {180}},\ \bibinfo {pages} {81} (\bibinfo {year}
  {2003})}\BibitemShut {NoStop}%
\bibitem [{\citenamefont {Lootens}\ \emph {et~al.}(2022)\citenamefont
  {Lootens}, \citenamefont {Vancraeynest-De~Cuiper}, \citenamefont {Schuch},\
  and\ \citenamefont {Verstraete}}]{lootens_mapping_2022}%
  \BibitemOpen
  \bibfield  {author} {\bibinfo {author} {\bibfnamefont {L.}~\bibnamefont
  {Lootens}}, \bibinfo {author} {\bibfnamefont {B.}~\bibnamefont
  {Vancraeynest-De~Cuiper}}, \bibinfo {author} {\bibfnamefont {N.}~\bibnamefont
  {Schuch}},\ and\ \bibinfo {author} {\bibfnamefont {F.}~\bibnamefont
  {Verstraete}},\ }\href {https://doi.org/10.1103/PhysRevB.105.085130}
  {\bibfield  {journal} {\bibinfo  {journal} {Physical Review B}\ }\textbf
  {\bibinfo {volume} {105}},\ \bibinfo {pages} {085130} (\bibinfo {year}
  {2022})},\ \bibinfo {note} {publisher: American Physical Society}\BibitemShut
  {NoStop}%
\bibitem [{\citenamefont {Etingof}\ \emph {et~al.}(2005)\citenamefont
  {Etingof}, \citenamefont {Nikshych},\ and\ \citenamefont
  {Ostrik}}]{etingof_fusion_2005}%
  \BibitemOpen
  \bibfield  {author} {\bibinfo {author} {\bibfnamefont {P.}~\bibnamefont
  {Etingof}}, \bibinfo {author} {\bibfnamefont {D.}~\bibnamefont {Nikshych}},\
  and\ \bibinfo {author} {\bibfnamefont {V.}~\bibnamefont {Ostrik}},\ }\href
  {https://doi.org/10.4007/annals.2005.162.581} {\bibfield  {journal} {\bibinfo
   {journal} {Annals of Mathematics}\ }\textbf {\bibinfo {volume} {162}},\
  \bibinfo {pages} {581} (\bibinfo {year} {2005})}\BibitemShut {NoStop}%
\bibitem [{\citenamefont {Wang}(2010)}]{wang_topological_2010}%
  \BibitemOpen
  \bibfield  {author} {\bibinfo {author} {\bibfnamefont {Z.}~\bibnamefont
  {Wang}},\ }\href {https://doi.org/10.1090/cbms/112} {\emph {\bibinfo {title}
  {Topological {Quantum} {Computation}}}},\ \bibinfo {series} {{CBMS}
  {Regional} {Conference} {Series} in {Mathematics}}, Vol.\ \bibinfo {volume}
  {112}\ (\bibinfo  {publisher} {American Mathematical Society},\ \bibinfo
  {year} {2010})\BibitemShut {NoStop}%
\bibitem [{\citenamefont {Natale}(2019)}]{natale_classification_2019}%
  \BibitemOpen
  \bibfield  {author} {\bibinfo {author} {\bibfnamefont {S.}~\bibnamefont
  {Natale}},\ }in\ \href {https://doi.org/10.1142/9789813272880_0050} {\emph
  {\bibinfo {booktitle} {Proceedings of the {International} {Congress} of
  {Mathematicians} ({ICM} 2018)}}}\ (\bibinfo  {publisher} {WORLD SCIENTIFIC},\
  \bibinfo {address} {Rio de Janeiro, Brazil},\ \bibinfo {year} {2019})\ pp.\
  \bibinfo {pages} {173--200}\BibitemShut {NoStop}%
\bibitem [{\citenamefont {Etingof}\ \emph {et~al.}(2015)\citenamefont
  {Etingof}, \citenamefont {Gelaki}, \citenamefont {Nikshych},\ and\
  \citenamefont {Ostrik}}]{etingof_tensor_2015}%
  \BibitemOpen
  \bibfield  {author} {\bibinfo {author} {\bibfnamefont {P.}~\bibnamefont
  {Etingof}}, \bibinfo {author} {\bibfnamefont {S.}~\bibnamefont {Gelaki}},
  \bibinfo {author} {\bibfnamefont {D.}~\bibnamefont {Nikshych}},\ and\
  \bibinfo {author} {\bibfnamefont {V.}~\bibnamefont {Ostrik}},\ }\href
  {https://doi.org/10.1090/surv/205} {\emph {\bibinfo {title} {Tensor
  categories}}},\ \bibinfo {series} {Mathematical surveys and monographs}\ No.\
  \bibinfo {number} {volume 205}\ (\bibinfo  {publisher} {American Mathematical
  Society},\ \bibinfo {address} {Providence, Rhode Island},\ \bibinfo {year}
  {2015})\BibitemShut {NoStop}%
\bibitem [{\citenamefont
  {Schollwoeck}(2011)}]{schollwoeckDensitymatrixRenormalizationGroup2011}%
  \BibitemOpen
  \bibfield  {author} {\bibinfo {author} {\bibfnamefont {U.}~\bibnamefont
  {Schollwoeck}},\ }\href {https://doi.org/10.1016/j.aop.2010.09.012}
  {\bibfield  {journal} {\bibinfo  {journal} {Annals of Physics}\ }\textbf
  {\bibinfo {volume} {326}},\ \bibinfo {pages} {96} (\bibinfo {year} {2011})},\
  \Eprint {https://arxiv.org/abs/1008.3477} {arXiv:1008.3477 [cond-mat]}
  \BibitemShut {NoStop}%
\bibitem [{\citenamefont {Cirac}\ \emph {et~al.}(2021)\citenamefont {Cirac},
  \citenamefont {{P{\'e}rez-Garc{\'i}a}}, \citenamefont {Schuch},\ and\
  \citenamefont {Verstraete}}]{ciracMatrixProductStates2021}%
  \BibitemOpen
  \bibfield  {author} {\bibinfo {author} {\bibfnamefont {J.~I.}\ \bibnamefont
  {Cirac}}, \bibinfo {author} {\bibfnamefont {D.}~\bibnamefont
  {{P{\'e}rez-Garc{\'i}a}}}, \bibinfo {author} {\bibfnamefont {N.}~\bibnamefont
  {Schuch}},\ and\ \bibinfo {author} {\bibfnamefont {F.}~\bibnamefont
  {Verstraete}},\ }\href {https://doi.org/10.1103/RevModPhys.93.045003}
  {\bibfield  {journal} {\bibinfo  {journal} {Reviews of Modern Physics}\
  }\textbf {\bibinfo {volume} {93}},\ \bibinfo {pages} {045003} (\bibinfo
  {year} {2021})}\BibitemShut {NoStop}%
\bibitem [{\citenamefont {Leutheusser}\ and\ \citenamefont
  {Liu}(2023)}]{leutheusser2023emergent}%
  \BibitemOpen
  \bibfield  {author} {\bibinfo {author} {\bibfnamefont {S.}~\bibnamefont
  {Leutheusser}}\ and\ \bibinfo {author} {\bibfnamefont {H.}~\bibnamefont
  {Liu}},\ }\href {https://doi.org/10.1103/physrevd.108.086020} {\bibfield
  {journal} {\bibinfo  {journal} {Physical Review D}\ }\textbf {\bibinfo
  {volume} {108}},\ \bibinfo {pages} {086020} (\bibinfo {year}
  {2023})}\BibitemShut {NoStop}%
\bibitem [{\citenamefont {Witten}(2022)}]{witten2022crossed_product}%
  \BibitemOpen
  \bibfield  {author} {\bibinfo {author} {\bibfnamefont {E.}~\bibnamefont
  {Witten}},\ }\bibfield  {journal} {\bibinfo  {journal} {Journal of High
  Energy Physics}\ }\textbf {\bibinfo {volume} {2022}},\ \href
  {https://doi.org/10.1007/jhep10(2022)008} {10.1007/jhep10(2022)008} (\bibinfo
  {year} {2022})\BibitemShut {NoStop}%
\bibitem [{\citenamefont {Chandrasekaran}\ \emph {et~al.}(2023)\citenamefont
  {Chandrasekaran}, \citenamefont {Longo}, \citenamefont {Penington},\ and\
  \citenamefont {Witten}}]{chandrasekaran_algebra_2023}%
  \BibitemOpen
  \bibfield  {author} {\bibinfo {author} {\bibfnamefont {V.}~\bibnamefont
  {Chandrasekaran}}, \bibinfo {author} {\bibfnamefont {R.}~\bibnamefont
  {Longo}}, \bibinfo {author} {\bibfnamefont {G.}~\bibnamefont {Penington}},\
  and\ \bibinfo {author} {\bibfnamefont {E.}~\bibnamefont {Witten}},\ }\href
  {https://doi.org/10.1007/JHEP02(2023)082} {\bibfield  {journal} {\bibinfo
  {journal} {Journal of High Energy Physics}\ }\textbf {\bibinfo {volume}
  {2023}},\ \bibinfo {pages} {82} (\bibinfo {year} {2023})}\BibitemShut
  {NoStop}%
\bibitem [{\citenamefont {Longo}\ and\ \citenamefont
  {Witten}(2023)}]{longo2023continuous_entropy}%
  \BibitemOpen
  \bibfield  {author} {\bibinfo {author} {\bibfnamefont {R.}~\bibnamefont
  {Longo}}\ and\ \bibinfo {author} {\bibfnamefont {E.}~\bibnamefont {Witten}},\
  }\href {https://doi.org/10.4310/pamq.2023.v19.n5.a5} {\bibfield  {journal}
  {\bibinfo  {journal} {Pure and Applied Mathematics Quarterly}\ }\textbf
  {\bibinfo {volume} {19}},\ \bibinfo {pages} {2501} (\bibinfo {year}
  {2023})}\BibitemShut {NoStop}%
\bibitem [{\citenamefont {Akers}\ and\ \citenamefont
  {Sorce}(2024)}]{akers_relative_2024}%
  \BibitemOpen
  \bibfield  {author} {\bibinfo {author} {\bibfnamefont {C.}~\bibnamefont
  {Akers}}\ and\ \bibinfo {author} {\bibfnamefont {J.}~\bibnamefont {Sorce}},\
  }\href {https://doi.org/10.1103/PhysRevLett.133.201601} {\bibfield  {journal}
  {\bibinfo  {journal} {Physical Review Letters}\ }\textbf {\bibinfo {volume}
  {133}},\ \bibinfo {pages} {201601} (\bibinfo {year} {2024})},\ \bibinfo
  {note} {publisher: American Physical Society}\BibitemShut {NoStop}%
\bibitem [{\citenamefont {Jensen}\ \emph {et~al.}(2023)\citenamefont {Jensen},
  \citenamefont {Sorce},\ and\ \citenamefont
  {Speranza}}]{jensen_generalized_2023}%
  \BibitemOpen
  \bibfield  {author} {\bibinfo {author} {\bibfnamefont {K.}~\bibnamefont
  {Jensen}}, \bibinfo {author} {\bibfnamefont {J.}~\bibnamefont {Sorce}},\ and\
  \bibinfo {author} {\bibfnamefont {A.~J.}\ \bibnamefont {Speranza}},\ }\href
  {https://doi.org/10.1007/JHEP12(2023)020} {\bibfield  {journal} {\bibinfo
  {journal} {Journal of High Energy Physics}\ }\textbf {\bibinfo {volume}
  {2023}},\ \bibinfo {pages} {20} (\bibinfo {year} {2023})}\BibitemShut
  {NoStop}%
\bibitem [{\citenamefont {Kudler-Flam}\ \emph {et~al.}(2023)\citenamefont
  {Kudler-Flam}, \citenamefont {Leutheusser}, \citenamefont {Rahman},
  \citenamefont {Satishchandran},\ and\ \citenamefont
  {Speranza}}]{kudler-flam_covariant_2023}%
  \BibitemOpen
  \bibfield  {author} {\bibinfo {author} {\bibfnamefont {J.}~\bibnamefont
  {Kudler-Flam}}, \bibinfo {author} {\bibfnamefont {S.}~\bibnamefont
  {Leutheusser}}, \bibinfo {author} {\bibfnamefont {A.~A.}\ \bibnamefont
  {Rahman}}, \bibinfo {author} {\bibfnamefont {G.}~\bibnamefont
  {Satishchandran}},\ and\ \bibinfo {author} {\bibfnamefont {A.~J.}\
  \bibnamefont {Speranza}},\ }\href {https://doi.org/10.48550/arXiv.2312.07646}
  {\bibinfo {title} {A covariant regulator for entanglement entropy: proofs of
  the {Bekenstein} bound and {QNEC}}} (\bibinfo {year} {2023}),\ \Eprint
  {https://arxiv.org/abs/2312.07646} {arXiv:2312.07646} \BibitemShut {NoStop}%
\bibitem [{\citenamefont {Chen}\ and\ \citenamefont
  {Penington}(2024)}]{chenClockJustWay2024}%
  \BibitemOpen
  \bibfield  {author} {\bibinfo {author} {\bibfnamefont {C.-H.}\ \bibnamefont
  {Chen}}\ and\ \bibinfo {author} {\bibfnamefont {G.}~\bibnamefont
  {Penington}},\ }\href {https://doi.org/10.48550/arXiv.2406.02116} {\bibinfo
  {title} {A clock is just a way to tell the time: Gravitational algebras in
  cosmological spacetimes}} (\bibinfo {year} {2024}),\ \Eprint
  {https://arxiv.org/abs/2406.02116} {arXiv:2406.02116 [gr-qc, physics:hep-th]}
  \BibitemShut {NoStop}%
\bibitem [{\citenamefont {Fewster}\ \emph {et~al.}(2025)\citenamefont
  {Fewster}, \citenamefont {Janssen}, \citenamefont {Loveridge}, \citenamefont
  {Rejzner},\ and\ \citenamefont
  {Waldron}}]{fewsterQuantumReferenceFrames2025}%
  \BibitemOpen
  \bibfield  {author} {\bibinfo {author} {\bibfnamefont {C.~J.}\ \bibnamefont
  {Fewster}}, \bibinfo {author} {\bibfnamefont {D.~W.}\ \bibnamefont
  {Janssen}}, \bibinfo {author} {\bibfnamefont {L.~D.}\ \bibnamefont
  {Loveridge}}, \bibinfo {author} {\bibfnamefont {K.}~\bibnamefont {Rejzner}},\
  and\ \bibinfo {author} {\bibfnamefont {J.}~\bibnamefont {Waldron}},\ }\href
  {https://doi.org/10.1007/s00220-024-05180-7} {\bibfield  {journal} {\bibinfo
  {journal} {Communications in Mathematical Physics}\ }\textbf {\bibinfo
  {volume} {406}},\ \bibinfo {pages} {19} (\bibinfo {year} {2025})}\BibitemShut
  {NoStop}%
\bibitem [{\citenamefont {De~Vuyst}\ \emph {et~al.}(2024)\citenamefont
  {De~Vuyst}, \citenamefont {Eccles}, \citenamefont {Hoehn},\ and\
  \citenamefont {Kirklin}}]{devuystGravitationalEntropyObserverdependent2024}%
  \BibitemOpen
  \bibfield  {author} {\bibinfo {author} {\bibfnamefont {J.}~\bibnamefont
  {De~Vuyst}}, \bibinfo {author} {\bibfnamefont {S.}~\bibnamefont {Eccles}},
  \bibinfo {author} {\bibfnamefont {P.~A.}\ \bibnamefont {Hoehn}},\ and\
  \bibinfo {author} {\bibfnamefont {J.}~\bibnamefont {Kirklin}},\ }\href
  {https://doi.org/10.48550/arXiv.2405.00114} {\bibinfo {title} {Gravitational
  entropy is observer-dependent}} (\bibinfo {year} {2024}),\ \Eprint
  {https://arxiv.org/abs/2405.00114} {arXiv:2405.00114} \BibitemShut {NoStop}%
\bibitem [{\citenamefont {Vuyst}\ \emph
  {et~al.}(2024{\natexlab{a}})\citenamefont {Vuyst}, \citenamefont {Eccles},
  \citenamefont {Hoehn},\ and\ \citenamefont
  {Kirklin}}]{vuystLinearizationInstabilitiesCrossed2024}%
  \BibitemOpen
  \bibfield  {author} {\bibinfo {author} {\bibfnamefont {J.~D.}\ \bibnamefont
  {Vuyst}}, \bibinfo {author} {\bibfnamefont {S.}~\bibnamefont {Eccles}},
  \bibinfo {author} {\bibfnamefont {P.~A.}\ \bibnamefont {Hoehn}},\ and\
  \bibinfo {author} {\bibfnamefont {J.}~\bibnamefont {Kirklin}},\ }\href
  {https://doi.org/10.48550/arXiv.2411.19931} {\bibinfo {title} {Linearization
  (in)stabilities and crossed products}} (\bibinfo {year}
  {2024}{\natexlab{a}}),\ \Eprint {https://arxiv.org/abs/2411.19931}
  {arXiv:2411.19931} \BibitemShut {NoStop}%
\bibitem [{\citenamefont {Vuyst}\ \emph
  {et~al.}(2024{\natexlab{b}})\citenamefont {Vuyst}, \citenamefont {Eccles},
  \citenamefont {Hoehn},\ and\ \citenamefont
  {Kirklin}}]{vuystCrossedProductsQuantum2024}%
  \BibitemOpen
  \bibfield  {author} {\bibinfo {author} {\bibfnamefont {J.~D.}\ \bibnamefont
  {Vuyst}}, \bibinfo {author} {\bibfnamefont {S.}~\bibnamefont {Eccles}},
  \bibinfo {author} {\bibfnamefont {P.~A.}\ \bibnamefont {Hoehn}},\ and\
  \bibinfo {author} {\bibfnamefont {J.}~\bibnamefont {Kirklin}},\ }\href
  {https://doi.org/10.48550/arXiv.2412.15502} {\bibinfo {title} {Crossed
  products and quantum reference frames: On the observer-dependence of
  gravitational entropy}} (\bibinfo {year} {2024}{\natexlab{b}}),\ \Eprint
  {https://arxiv.org/abs/2412.15502} {arXiv:2412.15502} \BibitemShut {NoStop}%
\bibitem [{\citenamefont {Araki}\ and\ \citenamefont
  {Woods}(1968)}]{araki_classification_1968}%
  \BibitemOpen
  \bibfield  {author} {\bibinfo {author} {\bibfnamefont {H.}~\bibnamefont
  {Araki}}\ and\ \bibinfo {author} {\bibfnamefont {E.~J.}\ \bibnamefont
  {Woods}},\ }\href {https://doi.org/10.2977/prims/1195195263} {\bibfield
  {journal} {\bibinfo  {journal} {Publications of the Research Institute for
  Mathematical Sciences}\ }\textbf {\bibinfo {volume} {4}},\ \bibinfo {pages}
  {51} (\bibinfo {year} {1968})}\BibitemShut {NoStop}%
\bibitem [{\citenamefont {Powers}(1967)}]{powers_representations_1967}%
  \BibitemOpen
  \bibfield  {author} {\bibinfo {author} {\bibfnamefont {R.~T.}\ \bibnamefont
  {Powers}},\ }\href {https://doi.org/10.2307/1970364} {\bibfield  {journal}
  {\bibinfo  {journal} {Annals of Mathematics}\ }\textbf {\bibinfo {volume}
  {86}},\ \bibinfo {pages} {138} (\bibinfo {year} {1967})}\BibitemShut
  {NoStop}%
\bibitem [{\citenamefont {Haagerup}(1985)}]{haagerup_injectivity_1985}%
  \BibitemOpen
  \bibfield  {author} {\bibinfo {author} {\bibfnamefont {U.}~\bibnamefont
  {Haagerup}},\ }in\ \href {https://doi.org/10.1007/BFb0074885} {\emph
  {\bibinfo {booktitle} {Operator Algebras and their Connections with Topology
  and Ergodic Theory}}},\ \bibinfo {editor} {edited by\ \bibinfo {editor}
  {\bibfnamefont {H.}~\bibnamefont {Araki}}, \bibinfo {editor} {\bibfnamefont
  {C.~C.}\ \bibnamefont {Moore}}, \bibinfo {editor} {\bibfnamefont
  {{\c{S}}.-V.}\ \bibnamefont {Stratila}},\ and\ \bibinfo {editor}
  {\bibfnamefont {D.-V.}\ \bibnamefont {Voiculescu}}}\ (\bibinfo  {publisher}
  {Springer Berlin Heidelberg},\ \bibinfo {address} {Berlin, Heidelberg},\
  \bibinfo {year} {1985})\ pp.\ \bibinfo {pages} {170--222}\BibitemShut
  {NoStop}%
\bibitem [{\citenamefont {Verstraete}\ and\ \citenamefont
  {Cirac}(2004)}]{verstraeteRenormalizationAlgorithmsQuantumMany2004}%
  \BibitemOpen
  \bibfield  {author} {\bibinfo {author} {\bibfnamefont {F.}~\bibnamefont
  {Verstraete}}\ and\ \bibinfo {author} {\bibfnamefont {J.~I.}\ \bibnamefont
  {Cirac}},\ }\href {https://doi.org/10.48550/arXiv.cond-mat/0407066} {\bibinfo
  {title} {Renormalization algorithms for {{Quantum-Many Body Systems}} in two
  and higher dimensions}} (\bibinfo {year} {2004}),\ \Eprint
  {https://arxiv.org/abs/cond-mat/0407066} {arXiv:cond-mat/0407066}
  \BibitemShut {NoStop}%
\bibitem [{\citenamefont {{Perez-Garcia}}\ \emph
  {et~al.}(2007{\natexlab{b}})\citenamefont {{Perez-Garcia}}, \citenamefont
  {Verstraete}, \citenamefont {Cirac},\ and\ \citenamefont
  {Wolf}}]{perez-garciaPEPSUniqueGround2007}%
  \BibitemOpen
  \bibfield  {author} {\bibinfo {author} {\bibfnamefont {D.}~\bibnamefont
  {{Perez-Garcia}}}, \bibinfo {author} {\bibfnamefont {F.}~\bibnamefont
  {Verstraete}}, \bibinfo {author} {\bibfnamefont {J.~I.}\ \bibnamefont
  {Cirac}},\ and\ \bibinfo {author} {\bibfnamefont {M.~M.}\ \bibnamefont
  {Wolf}},\ }\href {https://doi.org/10.48550/arXiv.0707.2260} {\bibinfo {title}
  {{{PEPS}} as unique ground states of local {{Hamiltonians}}}} (\bibinfo
  {year} {2007}{\natexlab{b}}),\ \Eprint {https://arxiv.org/abs/0707.2260}
  {arXiv:0707.2260 [cond-mat, physics:math-ph, physics:quant-ph]} \BibitemShut
  {NoStop}%
\bibitem [{\citenamefont {Blackadar}(2006)}]{blackadar_operator_2006}%
  \BibitemOpen
  \bibfield  {author} {\bibinfo {author} {\bibfnamefont {B.}~\bibnamefont
  {Blackadar}},\ }\href {https://doi.org/10.1007/3-540-28517-2} {\emph
  {\bibinfo {title} {Operator Algebras}}},\ \bibinfo {series} {Encyclopaedia of
  Mathematical Sciences}, Vol.\ \bibinfo {volume} {122}\ (\bibinfo  {publisher}
  {Springer},\ \bibinfo {year} {2006})\BibitemShut {NoStop}%
\end{thebibliography}%
\clearpage
\appendix

\section{Infinite tensor products} \label{sec:ITPFI}
\def\ITPFI{\mathrm{ITPFI}}

In the following, we make use of factors induced by infinite tensor products of finite type I factors, so-called ITPFI factors. In our context, we can think of them as the von Neumann algebras induced by infinite spin chains, with reference state $\ox_{i\in \ZZ} \rho_{A,i}$. Via purification, we can think of them as Alice's half of two half-infinite spin chains next to each other where the $i$-th spin of Alice and the $i$-th spin of Bob are in an entangled pure state $\ket{\Psi_i}$ that purifies $\rho_{A,i}$, see \cref{fig:itpfi}.
We denote the resulting factor by $\ITPFI(\{\rho_{A,i}\})$ and by $\ITPFI(\rho_A)$ if $\rho_{A,i}=\rho_A$ for all $i$. The type of $\ITPFI(\{\rho_{A,i}\})$ solely depends on the spectra of $\rho_{A,i}$ as $i\rightarrow \pm \infty$ \cite{araki_classification_1968}.  Every approximately finite-dimensional (AFD) factor not of finite type $\I$ and not of type $\III_0$ is isomorphic to $\ITPFI(\rho)$ for suitable $\rho$ \cite{powers_representations_1967,connes_classification_1976, haagerup_injectivity_1985}.  In particular, the unique AFD type $\II_1$ factor results whenever $\rho$ is maximally mixed, see \cite[Cor.~XIV.1.12]{takesaki3}.
If each $\rho_{A,i}$ is pure we have $\ITPFI(\{ \rho_{A,i}\}) = \B(\H)$, i.e., a type $\I$ factor. 
In particular, if $\rho_{A,i}=\tr_B(\ketbra{\Psi_i}{\Psi_i}), \rho_{B,i} = \tr_A(\ketbra{\Psi_i}{\Psi_i})$ we have
that $\M_A := \ITPFI(\{\rho_{A,i}\})$ and $\M_B:= \ITPFI(\{ \rho_{B,i})\}$ are commuting factors in Haag duality on a Hilbert space $\H$ with $\B(\H) = \ITPFI(\{\ketbra{\Psi_i}{\Psi_i}\})$. 
The construction naturally generalizes to any finite number of parties.

\begin{figure}[t!]
    \centering
    \includegraphics[width=0.9\linewidth]{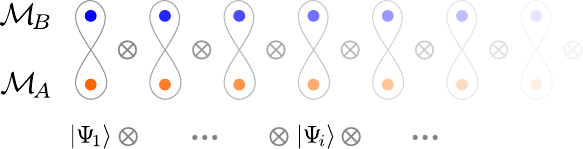}
    \caption{Illustration of ITPFI factors: Two half-infinite spin chains described by an infinite tensor product of bipartite entangled pure states $\ket{\Psi_i}$ (with marginals $\rho_{A,i}$ and $\rho_{B,i}$, respectively) give rise to a (infinite-dimensional) type $\I$ factor $\B(\H)$ with commuting subfactors $\M_A$ and $\M_B$ described by the ITPFI factors $\M_A = \ITPFI(\{\rho_{A,i}\})$ and $\M_B = \ITPFI(\{\rho_{B,i}\})$ which fulfill Haag duality, $\M_A'=  \M_B$.}
    \label{fig:itpfi}
\end{figure}

\section{Details of the construction in $D\geq 2$}
This appendix provides the proofs for the statements regarding the construction of gapped ground states of variable types in $D\geq 2$ sketched in the main text, see also \cref{fig:construction} for reference. 
The construction follows that of (trivial) projected entangled-pair states (PEPS) and works as follows: For concreteness, we consider a hypercubic lattice\footnote{Any other lattice with constant degree $d$ works by simply replacing $2D$ with $d$ in the following.} as an infinite graph with vertex set $\mc V=\ZZ^D$, where next neighbours are connected by an edge $e\in \mc E$. Each site $v\in \mc V$ is described by a Hilbert space of dimension $m^{2D}$. That is, each local site consists of $2D$ virtual spins of dimension $m$. 
We associate a pair of virtual spins with each edge $e\in\mc E$ such that each virtual spin is part of exactly one edge (see \cref{fig:construction}). Labelling each local virtual spin by $j\in \{1,\ldots,2D\}$ we also define the set of virtual spins $\tilde{\mc V}=\mc V\times \{1,\ldots,2D\}$.   Each virtual spin $\tilde v\in \tilde{\mc V}$ is part of a unique edge $e(\tilde v)$.
Consider a density matrix $\rho$ on $\CC^m$ and let $\ket\psi \in \CC^m \otimes \CC^m$ be its canonical purification (which is invariant under swap of the two virtual spins). We write $\ket\psi_e$ for a copy of $\ket\psi$ placed on edge $e\in \mc E$. 
The ground state of the Hamiltonian is then constructed as follows: 
With any finite subset $\Lambda \subset \ZZ^D$, we associate the set of edges $\mc E_\Lambda$ containing all edges within $\Lambda$. On the local algebra $\A_\Lambda$ we then define the pure state
\begin{align}
     \ket\Omega_{\rho,\Lambda} = \bigg(\bigotimes_{e\in \mc E_\Lambda}\ket\psi_e\bigg)\otimes \bigg(\bigotimes_{\tilde v\in \tilde{\mc V}, e(\tilde v)\notin \mc E_\Lambda} \ket 1_{\tilde v}\bigg).
\end{align}
We obtain a unique state $\omega$ on the local algebra $\A_{\ZZ^d}$ by setting
\begin{align}
    \omega_\rho(a) = \lim_{\Lambda \nearrow \ZZ^D} \bra \Omega_{\rho,\Lambda} a \ket\Omega_{\rho,\Lambda}.
\end{align}
The state $\omega_\rho$ may be represented as a an (injective) PEPS with finite bond-dimension \cite{verstraeteRenormalizationAlgorithmsQuantumMany2004,perez-garciaPEPSUniqueGround2007}.
If we write $P_e(\rho)$ for the projector $|\psi\rangle\langle\psi|_e \otimes 1$, then $\omega$ is the unqiue ground state of the (formal) Hamiltonian
\begin{align}
    H(\rho) = -\sum_{e} P_e(\rho).
\end{align}
Since $[P_e(\rho),P_{e'}(\rho)]=0$ for all $e\in\mc E$, it is a gapped, frustration-free, translation-invariant, nearest neighbor, commuting projector Hamiltonian. 
The gap is independent of $\rho$. If $s\mapsto \rho(s)$ is a path of density matrices, then the induced path $H(\rho(s))$ is local and gapped for all $s$. 
Clearly, there exist finite-depth quantum circuits mapping between the states $\omega_{\rho(s)}$. By choosing $\rho(0)=|1\rangle\langle 1|$, corresponding to a product state, we find that all the Hamiltonians $H(\rho)$ correspond to the trivial gapped phase of matter.

\subsection{The local observable algebras}
Consider an arbitrary infinite subset $A\subset \mc V$ and its complement $B$.
We denote by $\partial A\subset \mc E$ the set of edges not fully contained in $A$, but ending in $A$. For each $e\in \partial A$ there is precisely one virtual spin $\tilde v_A$ in $A$ and one virtual spin $\tilde v_B$ in $B$.
If $\Omega_\psi\in \mc H$ is the GNS representation of the ground state $\omega_{\psi}$, we can write it symbolically as 
\begin{align}
    \Omega_{\psi} &= \Omega_{\psi,\mathrm{int}(A)} \ox \Omega_{\psi, \partial A} \ox \Omega_{\psi, \mathrm{int}(B)} \\ 
    \H &= \H_{\mathrm{int}(A)}\ox  \H_{\partial A}\ox \H_{\mathrm{int}(B)},
\end{align}
where $\mathrm{int}(A)\subset \tilde{\mc V}$ contains all virtual spins that are not part of any edge in $\partial A$ and $\Omega_{\psi,\partial A} = \otimes_{e\in\partial A}\ket\psi_e = \Omega_{\psi,\partial B}$. 
It follows that the von Neumann algebras $\M_A,\M_B$ associated with $A$ and $B$ are factors in Haag duality such that
\begin{align}
    \M_{A/B} \cong \B(\H_{\mathrm{int}(A)}) \ox \ITPFI(\rho).
\end{align}
Due to the factor $ \B(\H_{\mathrm{int}(A)})$, $\M_{A/B}$ are properly infinite (hence cannot be of type $\I_n$ or $\II_1$). Their subtype is fully specified by the spectrum of $\rho$. 
In particular, if $m\geq 3$, $\ITPFI(\rho)$ can have type $\I_\infty,\II_1,\III_\lambda$ with $0<\lambda\leq 1$. Since $\I_\infty \ox \II_1 \cong \II_\infty$, the claim on the local types follows.

\section{Stability of approximately finite-dimensional factors}\label{app:stability}

Let us sketch why every approximately finite-dimensional (AFD) factor $\M$ not of finite type $\I$ is stable. If $\mc K$ is  an $n$-dimensional Hilbert-space we have $\B(\K) \cong M_n(\CC)$.  We thus want to show that $\M\cong \M \ox M_n(\CC)$ for every $n\in \NN$.

AFD factors are completely classified \cite{connes_classification_1976,haagerup_connes_1987}, and we can separate two cases: either $\M$ is finite, which implies that $\M$ is the unique AFD type $\II_1$ factor, or $\M$ is \emph{properly infinite}, 
which means that for any $n\in \NN$ we can find $n$ orthogonal projections $p_i = v_i v_i^\dagger$ such that $\sum_i p_i = 1$ and where $v_i$ are isometries such that 
\begin{align}
\label{eq:cuntz}
	v_i^\dagger v_j = \delta_{ij} 1,\quad \sum_{i} v_i v_i^\dagger =1. 
\end{align}
The stability of the $\II_1$ factor follows from the fact that it is isomorphic to any ITPFI factor given by an infinite tensor product of maximally mixed states \cite[Cor.~XIV.2.12]{takesaki3}.

For the properly infinite case, we consider the unitary operator $u: \H\ox \CC^n \rightarrow \H$ defined in terms of the isometries $v_{i}$ by
\begin{align}
	u = \sum_i v_i \ox \bra i.
\end{align}
From \cref{eq:cuntz}, we find that for any operator of the form $a\ox \ketbra{i}{j}\in \M\ox M_n(\CC)$ we have
\begin{align}
	u (a\ox \ketbra{i}{j})u^\dagger = v_i a v_j^\dagger \in \M.
\end{align}
Hence $u$ realizes the isomorphism $\M\cong \M\ox M_n(\CC)$.

\section{Proof of \cref{prop:stable-type} and \cref{thm:stable-bipartite}}
We begin by noting that if $A$ is infinite, then the factor $\M_A$ cannot have finite type $\I$. This follows because the quasi-local C* algebra generated by $\A_A$ (with respect to the usual operator norm on matrices) clearly does not allow for (non-trivial) finite-dimensional representations. Therefore $\M_A$ is stable,  $\M_A \cong \M_A \ox M_n(\CC)$ for any $n\in \NN$. 

\begin{proof}[Proof of \cref{prop:stable-type}]

We may assume that $R$ is disjoint from $A$. Now,  note that $\M_{A\cup R}$ is the von Neumann algebra generated by $\M_A$ and the finite type $\I$ factor $\M_R$, which commutes with $\M_A$. Therefore, we have a factorization $\H=\H_{R^c}\ox \H_R$, with $\H_R$ finite-dimensional, such that $\M_R = 1\ox \B(\H_R)$ and, for $A\subset R^c$, there is a unique von Neumann algebra $\N_A\cong \M_A$ on $\H_{R^c}$ such that $\M_A=\N_A\ox 1\subset\B(\H_{R^c})\ox 1$. Thus, the stability of AFD factors that are not of finite type $\I$ (see \cref{app:stability}) implies the first part of \cref{prop:stable-type}: $\M_{A\cup R} = \N_A \ox \B(\H_R)\cong \N_A \cong \M_A$.
It remains to show that Haag duality holds for $\M_{A\cup R}$ if it holds for $\M_A$.
\begin{align*}
    \M_{A\cup R}&= (\N_A\ox\M_R)' \\
    &= (\N_A\ox \B(\H_R))' \\
    &= \N_A' \ox \B(\H_R)' \\
    &= \N_{(A\cup R)^c} \ox 1 = \M_{(A\cup R)^c}.\qedhere
\end{align*}    
\end{proof}

\begin{lemma}\label{lem:properly-inf}
    Suppose that $\H$ contains a translation invariant vector $\Omega\ne0$ and that $A$ is a properly infinite region. 
    Then, the factor $\M_A$ is properly infinite.
\end{lemma}
\begin{proof}
    Since $A$ is (properly) infinite, $\M_A$ clearly cannot have finite type $\I$ (the quasi-local algebra induced by $\A_A$ does not have non-trivial finite-dimensional representations). It thus remains to show that $\M_A$ does not have type $\II_1$. This can be shown via the same argument as in \cite[Thm.~11.1.2]{naaijkensAnyonsInfiniteQuantum2012}, which adapts the argument of \cite[Prop.~5.3]{keyl_entanglement_2006}. To this end, we note that $A$ being properly infinite implies that for any finite region $R\Subset \Gamma$, there is a ball of arbitrary radius in $A\setminus (R\cap A)$.
\end{proof}

\begin{lemma}\label{lem:properly-inf2}
    Suppose that $\H$ contains a translation invariant vector $\Omega\ne0$.
    Let $A$ be a properly infinite region such that $A^c$ is properly infinite as well and assume that Haag duality for the bipartition $\Gamma = A\cup A^c$.
    Then $\M_A$ is in standard representation on $\H$.
\end{lemma}

\begin{proof}
    The result follows from \cref{lem:properly-inf} since a properly infinite von Neumann algebra $\N$ on $\H$ is in standard representation if and only if $\N'$ is properly infinite as well \cite[Thm.~III.2.6.16]{blackadar_operator_2006}.
\end{proof}

\begin{proof}[Proof of \cref{thm:stable-bipartite}]
    \Cref{lem:properly-inf2} shows that both $\M_A$ and $\M_{\tilde A}$ are in standard representation.
    Since $A = \tilde A \cup R$ with $R$ finite, we know that 
    \begin{equation}\label{eq:budden}
        \M_{A} \cong \M_{\tilde A} \ox M_{d^{|R|}}(\CC),
    \end{equation}
    where we used $\M_R\cong M_{d^{|R|}}(\CC)$.
    Since $\M_A$ is an approximately finite dimensional factor that is not a finite type $\I$ factor, \eqref{eq:budden} implies that $\M_A\cong \M_{\tilde A}$.
    Since the standard representation of a von Neumann algebra is unique up to unitary equivalence, this implies the existence of a unitary $U$ on $\H$ such that $U\M_A U^\dagger = \M_{\tilde A}$.
    Since Haag duality holds, this further implies that $U\M_B U^\dagger = \M_{\tilde B}$.
\end{proof}

\section{Proof of \Cref{prop:distillation}}
\label{app:proof-distillation}

Here, we prove \cref{prop:distillation}. 
Consider a local Hamiltonian $H$ on the lattice $\Gamma$. Let $\Gamma_n$ be an increasing sequence of finite subsets such that $\cup_n \Gamma_n = \Gamma$.
In the main text, we choose $\Gamma_n$ to be balls of increasing radii, but that is not important.
Denote by $H_n$ the (formal) restriction of $H$ to $\Gamma_n$, retaining only those terms completely supported within $\Gamma_n$, and denote by $\Omega_n\in \H_n$ a ground state of $H_n$, where we introduced the (finite-dimensional) Hilbert-space $\H_n$ associated with $\Gamma_n$.
We can extend the state $\Omega_n$ by a fixed infinite product state, e.g., $\ket 0$ on every site in $\Gamma\setminus \Gamma_n$, to view it as a state on the full local algebra $\A_\Gamma$.
If the Hamiltonian $H$ on $\A_\Gamma$ has a unique ground state on $\A_\Gamma$, it is a general result that the sequence of finite-size ground states weak*-converges to the ground state in the thermodynamic limit. This follows from \cite[Prop. 5.3.25]{BR2} and the weak*-compactness of the state space. Hence for every $a\in A_\Gamma$ and $\eps>0$ there exists $n_0\in\NN$ such that for all $n\geq n_0$
\begin{align}\label{eq:app:weak-star}
  |\bra{\Omega_n}a\ket{\Omega_n} - \bra\Omega a\ket\Omega| <\eps.
\end{align}
Now assume that $\Phi \in (\CC^n)^{\otimes N}$ can be distilled from $\Omega$ relative to a fixed $N$-partition $A_1,\ldots,A_N$. 
Hence for every $\eps>0$ there exists an LOCC protocol $T_\eps$, such that
\begin{align}\label{eq:app:locc1}
    \bra{\Phi} \Tr_\H T_\eps(\kettbra\Omega \ox \kettbra 0^{\otimes N})\ket{\Phi}   
    \geq 1-\eps.
\end{align} 
The observable algebra of the finite-dimensional $N$-partite system, on which $\Phi$ is a pure state, is $\B=M_n(\CC)^{\ox N}$. For ease of notation we denote by $\B_j\cong M_n(\CC)\subset \B$ the $j$-th factor of $\B$.
The overall quantum channel $T_\eps$ associated with an LOCC protocol can be represented as
\begin{align}
    T_\eps = \sum_{\alpha=1}^M k_\alpha(\placeholder)k_\alpha^\dagger, \qquad k_\alpha = k_\alpha\up{1} \cdots k_\alpha\up{N},
\end{align}
where each $k_\alpha\up{j}$ is an element of $\A_{A_j} \ox \B_j$ and
where $M\in\NN$ ($M$ encodes the total amount of classical communication throughout the full LOCC protocol).

Let us define the operators
\begin{align}
    \tilde k_\alpha = (1\ox \bra\Phi) k_\alpha(1\ox \ket 0^{\ox N}) \in \A_\Gamma.
\end{align}
We can now write \eqref{eq:app:locc1} as
\begin{align}
    \sum_{\alpha=1}^M \Tr_\H(\tilde k_\alpha \kettbra\Omega \tilde k_\alpha^\dagger) = \sum_{\alpha=1}^M \bra\Omega \tilde k_\alpha^\dagger \tilde k_\alpha \ket\Omega \geq 1-\eps.\nonumber
\end{align}
Since each $\tilde k_\alpha$ has finite support, the support of $a:= \sum_{\alpha} \tilde k_\alpha^\dagger \tilde k_\alpha$ is finite, too.
We denote its support by $\tilde \Gamma_\eps$. 
By weak*-convergence, there exists an $n_0$ such that for all $n\geq n_0$ we have $\tilde \Gamma_\eps \subset \Gamma_n$ and
\begin{align}
    |\bra{\Omega_n} a \ket{\Omega_n} - \bra\Omega a \ket\Omega | \leq \eps.\nonumber
\end{align}
Hence we find
\begin{align}
     \bra{\Phi} &\Tr_{\H_n} T_\eps(\kettbra{\Omega_n} \ox \kettbra 0^{\otimes N})\ket{\Phi} \nonumber\\
     &= \sum_\alpha \Tr_{\H_n}(\tilde k_\alpha \kettbra{\Omega_n} \tilde k_\alpha^\dagger)\nonumber \\
    &= \bra{\Omega_n}a \ket{\Omega_n}  
     \geq \bra{\Omega}a\ket{\Omega} -\eps  
     \geq 1-2\eps.
\end{align}
Thus, for every $\eps>0$, there exists an $n_0$ such that for all $n\geq n_0$ $\Phi$ may be distilled from $\ket{\Omega_n}$ up to error $\eps$, which finishes the proof.
\end{document}